\documentclass[10pt]{article}
\usepackage{graphicx}

\usepackage{amsmath}
\usepackage{mathrsfs}
\usepackage{latexsym}
\usepackage{amssymb}
\usepackage{amsfonts}
\usepackage{pifont}
\usepackage{epsfig}
\usepackage{array}
\usepackage{color}
\usepackage{enumitem}
\usepackage{theorem}

\numberwithin{equation}{section}

\def\E{{\mathbb  E}}
\def\Pr{{\mathbb P}}

\def\eqdef{\triangleq}

\def\ind{{\bf 1}}

\def\A{\mathcal{A}}

\def\T{\mathcal{T}}
\newtheorem{theorem}{Theorem}
\newtheorem{proposition}{Proposition}
\newtheorem{corollary}{Corollary}
\newtheorem{remark}{Remark}
\newtheorem{definition}{Definition}

\makeatletter
\def\blfootnote{\gdef\@thefnmark{}\@footnotetext}
\makeatother

 \usepackage{mathptmx}      
\begin{document}

\title{Instability of Sharing Systems in the Presence of Retransmissions}

\author{Predrag R. Jelenkovi\'c  \hspace{0.5cm}    Evangelia D. Skiani \\
\begin{tabular} {c}
\small Department of Electrical Engineering \\
\small Columbia University, New York, NY 10027 \\
\small \{predrag, valia\}@ee.columbia.edu \\
\end{tabular} }

\date{\today}

\maketitle

\begin{abstract}
Retransmissions represent a primary failure recovery mechanism on all layers of communication network architecture. Similarly, fair sharing, e.g. processor sharing (PS), is a widely accepted approach to resource allocation among multiple users. Recent work has shown that retransmissions in failure-prone, e.g. wireless ad hoc, networks can cause heavy tails and long delays. In this paper, we discover a new phenomenon showing that PS-based scheduling induces complete instability with zero throughput in the presence of retransmissions, regardless of how low the traffic load may be. This phenomenon occurs even when the job sizes are bounded/fragmented, e.g. deterministic. Our analytical results are further validated via simulation experiments. Moreover, our work demonstrates that scheduling one job at a time, such as first-come-first-serve, achieves stability and should be preferred in these systems.
\vspace{2mm}

\end{abstract}

\blfootnote{This work is supported by NSF Grant number 0915784.}
\blfootnote{Preliminary version of this paper has appeared earlier in SIGMETRICS'14 \cite{JS2014}.}

\section{Introduction}\label{s:intro}
High variability and frequent failures characterize the majority of large-scale systems, e.g. infrastructure-less wireless networks, cloud/parallel computing systems, etc. The nature of these systems imposes the employment of failure recovery mechanisms to guarantee their good performance. One of the most straightforward and widely used recovery mechanism is to simply restart all the interrupted jobs from the beginning after a failure occurs. In communication systems, restart mechanisms lie at the core of the network architecture where retransmissions are used on all protocol layers to guarantee data delivery in the presence of channel failures, e.g. Automatic Repeat reQuest (ARQ) protocol \cite{BG92}, contention based ALOHA type protocols in the medium access control (MAC) layer, end-to-end acknowledgements in the transport layer, HTTP downloading scheme in the application layer, and others.

Furthermore, sharing is a primary approach to fair scheduling and efficient management of the available resources. Fair allocation of the network resources among different users can be highly beneficial for increasing throughput and utilization. For instance, CDMA is a multiple access method used in communication networks, where several users can transmit information simultaneously over a single channel via sharing the available bandwidth. Another example is Processor Sharing (PS) scheduling \cite{Ya07} where the capacity is equally shared between multiple classes of customers. In \emph{Generalized} PS (GPS) \cite{PG93,PG94}, service allocation is done according to some fixed weights. The related \emph{Discriminatory} PS (DPS) \cite{AEAK06,FMI80,KL67} is used in computing to model the Weighted Round Robin (WRR) scheduling, while it is also used in communications, as a flow level model of heterogenous TCP connections. Similarly, fair queuing (FQ) is a scheduling algorithm where the link capacity is fairly shared among active network flows; in weighted fair queuing (WFQ), which is the discretized version of GPS, different scheduling priorities are assigned to each flow.

In general, PS-based scheduling disciplines have been widely used in modeling computer and communication networks. Early investigations of PS queues were motivated by applications in multiuser computer systems \cite{Cof70}. The M/G/1 PS queue has been studied extensively in the literature \cite{Ya92}. In the case of the M/M/1 PS system, the conditional Laplace transform of the waiting time was derived in \cite{Cof70}. The importance of scheduling in the presence of heavy tails was first recognized in \cite{Ana99}, and later, in \cite{JM03}, the M/G/1 PS queue was studied assuming subexponential job sizes; see also \cite{JM03} for additional references.  

In \cite{WB12}, it was proven that, although there are policies known to optimize the sojourn time tail under a large class of heavy-tailed job sizes (e.g. PS and SRPT) and there are policies known to optimize the sojourn time tail in the case of light-tailed job sizes, e.g. FCFS, no policies are known to optimize the sojourn time tail across both light and heavy-tailed job size distributions. Indeed, such policies must ``learn" the job size distribution in order to optimize the sojourn time tail. In the heavy-tailed scenarios, any scheduling policy that assigns the server exclusively to a very large job, e.g. FCFS, may induce long delays, in which case, sharing guarantees better performance.

In this paper, we study the effects of sharing on the system performance when restarts are employed in the presence of failures. We revisit the well-studied M/G/1 queue with failures and restarts and focus on the PS scheduling policy. We use the following generic model, which was first introduced in \cite{FS05} in the application context of computing. The system dynamics is described as a process $(A, \{A_n\}_{n\geq1})$, where $A_n$ correspond to the periods when the system is available. $(A, \{A_n\}_{n\geq1})$ is a sequence of i.i.d random variables, independent of the job sizes. In each period of time that the system is available, say $A_n$, we attempt to execute a job of random size $B$. If $A_n > B$, we say that the job is successfully completed; otherwise, we restart the job from the beginning in the following period $A_{n+1}$ when the channel is available.

With regard to retransmissions, it was first recognized in \cite{FS05,SL06} that restart mechanisms may result in heavy-tailed (power law) delays even if the job sizes and failure rates are light-tailed. In \cite{PJ07RETRANS}, it was shown that the power law delays arise whenever the hazard functions of the data and failure distributions are proportional. In the practically important case of bounded data units, a uniform characterization of the entire body of the retransmission distribution was derived in \cite{JS2012,JS2012L}, which allows for determining the optimal size of data units/fragments in order to alleviate the power law effect. Later, these results were extended to the case where the channel is highly correlated \cite{JS2013}, i.e. switches between states with different characteristics, and was proved that the delays are insensitive to the channel correlations and are determined by the `best' channel state.

In this paper, our main contributions are the following. First, we prove that the M/G/1 PS queue is always unstable, regardless of how light the load is and how small the job sizes may be, see Theorems~\ref{thm:0} and~\ref{thm:0-b} in Section~\ref{s:0}. This is a new phenomenon, since, contrary to the conventional belief, sharing the service even between very small deterministic jobs can render the system completely unstable when retransmissions/restarts are employed. This instability is strong, in the sense of system having zero throughput. The intuition is the following. If a large number of jobs arrives in a short period of time, then under the elongated service time distribution induced by sharing, coupled with retransmissions, the queue will keep accumulating jobs that will equally share the capacity, which further exacerbates the problem. Every time a failure occurs, the system resets and the service requirement for each job elongates as the queue size increases. The expected delay until the system clears becomes increasingly long and, consequently, the queue will continue to grow leading to instability. This result also applies to the \emph{Discriminatory} PS (DPS) queue, where the service is not shared equally but according to some fixed weights. Next, we remove the Poisson assumption and extend our results to general renewal arrivals in Section~\ref{s:1}. This demonstrates that instability arises from the interplay between sharing and retransmission/restart mechanisms, rather than any specific characteristics of the arrival process and/or service distribution.

We would also like to emphasize that job fragmentation cannot stabilize the system regardless of how small the fragments are made, since Theorem~\ref{thm:0-b} shows instability for any minimum job size $\beta >0$. Similarly, the system cannot be stabilized by checkpointing regardless of how small the intervals between successive checkpoints are chosen. In our experimental results, we make an interesting observation on the system behavior before it saturates. Initially, during the transient period, the queue appears as if it were stable and one would have difficulty predicting the forthcoming instability. Although it may occasionally accumulate a substantial number of jobs, it returns to zero and starts afresh like a stable queue. However, there exists a time when the queue reaches a critical size after which the service rate of the jobs reduces so much that neither of them can depart. Hence, as the queue continues to increase in size, the system becomes unstable. 

Next, in order to gain further insight into the system, we focus on the transient behavior and study the properties of the completion time of a finite number of jobs with no future arrivals. Specifically, we compare two work-conserving policies: scheduling one job at a time, e.g. FCFS, and PS. Overall, we discover that serving one job at a time exhibits uniformly better performance than PS; compare Theorems~\ref{thm:2} and \ref{thm:3}, respectively. Furthermore, under more technical assumptions, and for light-tailed job/failure distributions, we show that PS performs distinctly worse compared to the heavy-tailed ones, and that PS is always unstable.

From an engineering perspective, our results indicate that traditional approaches in existing systems may be inadequate in the presence of failures. This new phenomenon demonstrates the need of revisiting existing techniques to large-scale failure-prone systems, where PS-based scheduling may perform poorly. For example, since PS is unstable even for deterministic jobs, packet fragmentation, which is widely used in communications, cannot alleviate instabilities. Indeed, fragmentation can only postpone the time when the instability occurs, but cannot eliminate the phenomenon; see Example~1 in Section~\ref{s:5}. Therefore, serving one job at a time, e.g. FCFS, is highly advisable in such systems.

The paper is organized as follows. In Section \ref{s:intro}, we introduce the model along with the necessary definitions and notation. Next, in Section~\ref{s:0}, we present our main results on the M/G/1 queue, which are further generalized in Section~\ref{s:1}. Later, in Section~\ref{s:3}, we analyze the transient behavior of the system under two different scheduling policies, e.g. serving one job at a time and PS. Last, Section~\ref{s:5} presents our simulation experiments that validate our main theoretical findings, while Section~\ref{s:6} concludes the paper.

\section{Definitions and Notation}\label{s:intro} 
First, we provide the necessary definitions and notation assuming that the jobs are served individually. Consider a generic job of random size $B$ requesting service in a failure-prone system. Without loss of generality, we assume that the system is of unit capacity. Its dynamics is described as a process $(A, \{A_n\}_{n\geq1})$ of availability periods, where at the end of each period $A_n$, the system experiences a failure, as shown in Figure~\ref{fig:tc}. 

 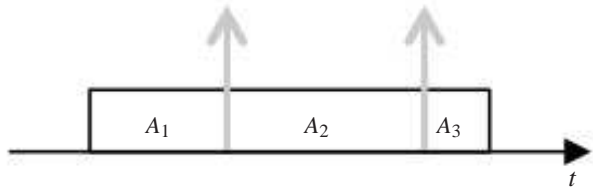
\begin{figure}[h]
\centering
\begin{picture}(220,80)(0,0)
\put(0,0){\includegraphics[scale = 1.7]{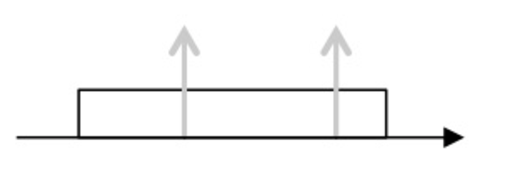}}
\put(60,25){\small $A_1$ }
\put(120,25){\small $A_2$}
\put(170,25){\small $A_3$}
\put(220,5){$t$}
\end{picture}
  \caption{System with failures.}
\label{fig:tc}
\end{figure}

At each period of time that the system becomes available, say $A_n$, we attempt to process a generic job of size $B$. If $A_n > B$, we say that the job is completed successfully; otherwise, we wait until the next period $A_{n+1}$ when the channel is available and restart the job. A sketch of the model depicting the system is drawn in Figure \ref{fig:tc2}.
 \begin{figure}[h]
\centering
\begin{picture}(220,70)(0,0)
\put(0,0){\includegraphics[scale = 0.4, trim = 10mm 70mm 30mm 50mm, clip ]{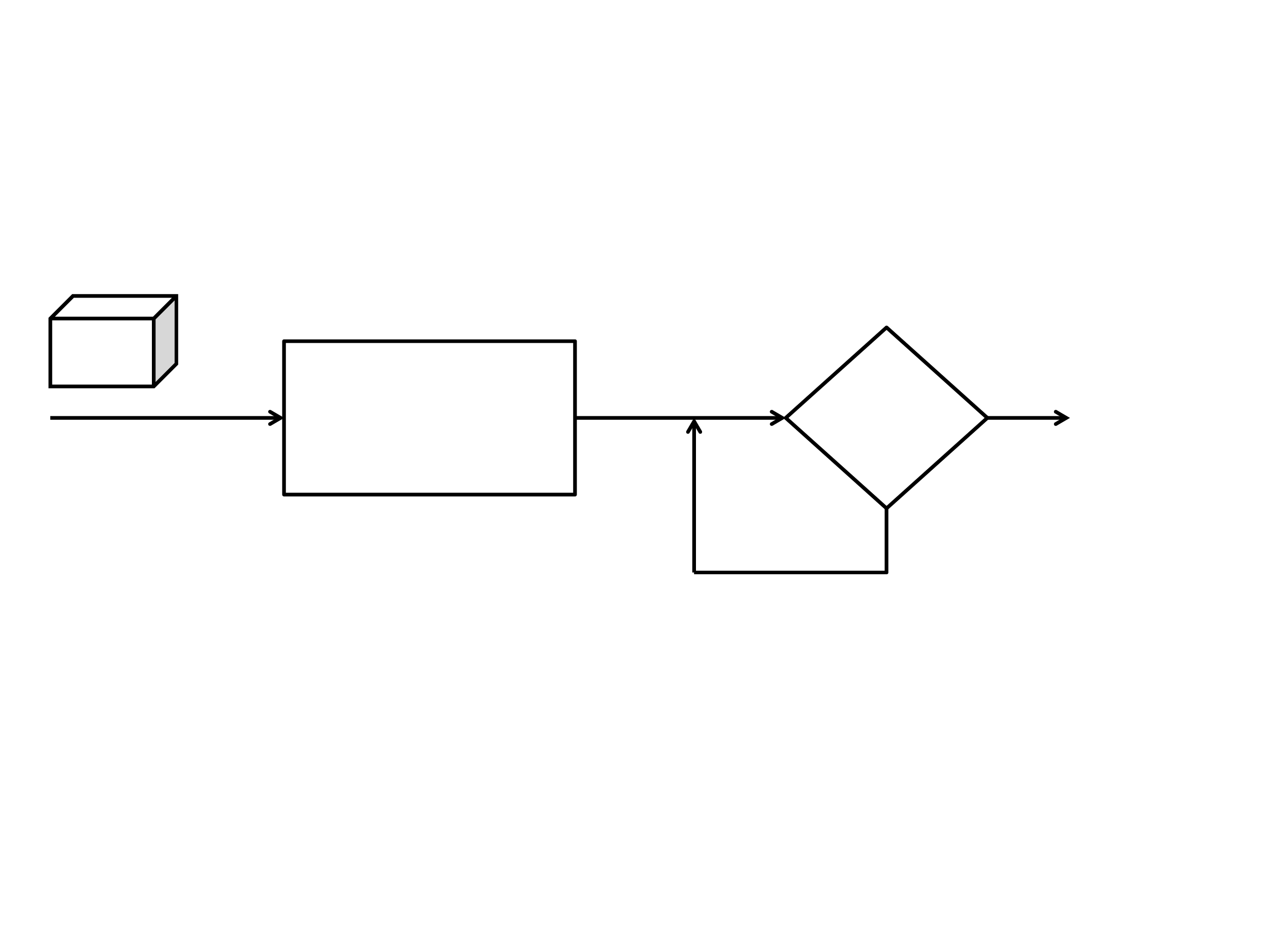}}
\put(9,55){\small $B$}
\put(60,50){\small Failure-prone }
\put(72,40){\small channel }
\put(74,30){\small $\{A_n\}$ }
\put(175,41){\small $A_n > B$}
\put(154,10){\small resend   $\qquad$      no}
\end{picture}
  \caption{Jobs executed in system with failures.}
\label{fig:tc2}
\end{figure}
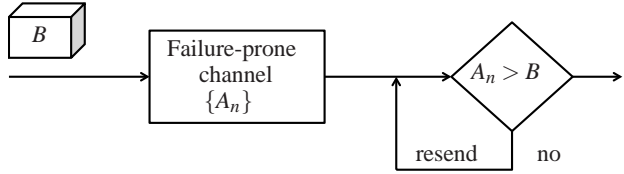

\begin{definition}
The number of restarts for a generic job of size $B$ is defined as
\begin{displaymath}
N \eqdef \inf \{  n: A_n > B\}.
\end{displaymath}
\end{definition}

We are interested in computing the total service time $S$ until $B$ is successfully completed, which is formally defined as follows. 

\begin{definition}\label{def:S}
The service time is the total time until a generic job of size $B$ is successfully completed and is denoted as
\begin{equation*} \label{eq:S}
S \eqdef \sum_{i=1}^{N-1} A_i + B.
\end{equation*}
\end{definition}

\noindent We  denote the complementary cumulative distribution functions for $A$ and
 $B$, respectively, as
  \begin{equation*}
 \bar G (x) \eqdef \Pr (A>x) \qquad \text{and} \qquad \bar{F}(x)\eqdef \Pr(B>x).
   \end{equation*}
Throughout the paper, we assume that the functions $\bar G(x)$ and $\bar F(x)$ are absolutely continuous for all $x \geq 0$. We also use the following standard notation. For any two real functions $f(x)$ and $g(x)$ and fixed $x_0 \in \mathbb{R} \cup \{ \infty\}$, we say $f(x) \sim g(x)$ as $ x \rightarrow x_0$, to denote~$\lim_{ x \rightarrow x_0} f(x)/g(x) = 1$. 

\section{M/G/1 queue with restarts}\label{s:0}
In this section, we discuss the stability of the M/G/1 queue under two scheduling disciplines: Processor Sharing (PS) and First Come First Serve (FCFS). Throughout the paper, we assume that the arrival rate is positive, $\lambda >0$, unless otherwise indicated. In the following subsection, we show in Theorem~\ref{thm:0-b} that the M/G/1 PS queue is unstable. Next, in subsection~\ref{ss:0-1}, we derive the necessary and sufficient condition for the system to be stable when the jobs are processed according to FCFS.   

\subsection{Instability of Processor Sharing Queue}\label{ss:0-2}
In this section, we show in Theorems~\ref{thm:0} and \ref{thm:0-b} that the M/G/1 PS queue is unstable when jobs need to restart after failures. First, in Proposition~\ref{prop:1}, we show that for some initial condition on the queue size, the probability that no job completes service approaches 1, under the mild assumption that jobs are bounded from below by some positive constant $\beta$. This is a natural assumption for communication or computing applications where jobs, e.g. files, packets, threads, must have a header to contain the required information, such as destination address, thread id, etc. Hence, the job sizes, in practice, cannot be smaller than a positive constant. 

Next, in Theorem~\ref{thm:0}, without any initial condition on the queue size, we prove that after some finite time, no job ever leaves the system; this result is stronger than standard stability theorems since it implies zero throughput. Then, in Corollary~\ref{cor:0}, we draw the weaker conclusion that the queue size grows to infinity, which is also stated in Theorem~\ref{thm:0-b}. Nevertheless, the latter does not require the assumption on the minimum job size.

We begin with the following proposition.
\begin{proposition}\label{prop:1}
Assume that at time $t=0$, a failure occurs and there are $Q_0 \geq k$ jobs in the M/G/1 PS queue. If $\E A < \infty$ and $\Pr [B \geq \beta] = 1, \beta >0$, then there exists $\theta>0$, such that for all $k \geq 1$
\begin{equation} \label{eq:prop1}
\Pr[\text{no job ever completes service}] \geq 1 - O(\E A \ind (A \geq \beta k ) + e^{-\theta k}).
\end{equation}
\end{proposition}

\begin{proof}
Let $T_1 = \sum_{i=1}^{c k} A_i$ be the cumulative time that includes the first $ck$ failures; to simplify notation we write $\sum_{x}^y $ to denote $\sum_{\lceil x \rceil}^{\lfloor y \rfloor} $, where $\lceil x \rceil$ is the smallest integer $\geq x$ and $\lfloor y \rfloor$ is the largest integer $\leq y$. Now, define the event $\mathcal{A}_1 \equiv \mathcal{A}_1(k) \eqdef \{ A_1 < \beta  k, A_2 <  \beta  k, \dots, A_{ck} < \beta  k \}$. On this event, no job can leave the system since $Q_0 \geq k$ and all of them are at least of size $\beta$. Thus, if they were served in isolation, they could not have completed service in the first $ck$ attempts.

Now, let $E_1$ denote the event that there is no departure in the first $ck$ attempts and there are at least $k$ arrivals in $(0, T_1]$; we use $Z_{(t_0, t_1]}$ to denote the number of Poisson arrivals in the interval $(t_0, t_1]$, whereas we simply write $Z_t$ for intervals $(0,t]$. Formally, 
\begin{equation*}
E_1 \supset \underline{E}_1 \eqdef \{ Z_{T_1} \geq k, \mathcal{A}_1 \},
\end{equation*}
on the set $\{ Q_0 \geq k\}$. Now, observe that
\begin{align*}
 \Pr( \underline{E}_1) & \geq \Pr(Z_{T_1} \geq k, T_1 \geq 2 k/ \lambda ,  \mathcal{A}_1) \\
& \geq  \Pr(Z_{ 2k/ \lambda} \geq k, T_1 \geq 2 k/ \lambda ,  \mathcal{A}_1) \\
& \geq \Pr(Z_{ 2k/ \lambda} \geq k ) \Pr(  T_1 \geq 2 k/ \lambda ,  \mathcal{A}_1), \\
\intertext{since Poisson arrivals are independent of the failure process. Thus,}
\Pr( \underline{E}_1) & \geq  \Pr(Z_{ 2k/ \lambda} \geq k ) \left( \Pr( \mathcal{A}_1 ) - \Pr( T_1 < 2k/ \lambda) \right).
\end{align*}

First, note that 
\begin{align*}
 \Pr(Z_{ 2k/ \lambda} \geq k ) &  = 1 - \Pr(Z_{ 2k/ \lambda} < k )   = 1 - \Pr( 2k - Z_{ 2k/ \lambda} > k )  \\
 & \geq 1 -  e^{-\theta k} \E e^{\theta (2k - Z_{ 2k/ \lambda})} = 1 - e^{\theta k } \E e^{-\theta Z_{ 2k/ \lambda}},
 \intertext{by Cramer's bound for $\theta >0$. Next, observe that $ Z_{ 2k/ \lambda}$ is Poisson with rate $ 2 k $ and thus}
\Pr(Z_{ 2k/ \lambda} \geq k )  & \geq 1- e^{\theta k } e^{2(e^{-\theta}-1)k} = 1 - e^{-\theta_1 k},
\end{align*}
where $\theta_1 = 2(1-e^{-\theta})- \theta  >0 $, for $\theta$ small.

Second, observe that
 \begin{align*}
\Pr( T_1 < 2k/ \lambda)  &  = \Pr \left(\sum_{i=1}^{ck} A_i < 2k/ \lambda \right) = \Pr \left(\sum_{i=1}^{ck} (A_i - \E A) < 2k/ \lambda - ck \E A \right) \\
  &  \leq \Pr \left(\sum_{i=1}^{3k/ \lambda \E A} ( \E A - A_i) > k/ \lambda  \right),
 \intertext{by picking $c \eqdef  3 / (\lambda \E A) $. Now, let $X_i  \eqdef \E A - A_i$, which are bounded from above since $X_i \leq \E A < \infty$, from our main assumption. Therefore, Cramer's large deviation bound implies that}
 \Pr( T_1 < 2k/ \lambda)  &  \leq \Pr \left(\sum_{i=1}^{3k/ \lambda \E A} X_i  > k/ \lambda  \right) \leq H_2 e^{-\theta_2 k},
\end{align*}
for some $H_2, \theta_2 >0$.

Therefore, 
 \begin{align*}
\Pr( \underline{E}_1)  &  \geq (1- e^{-\theta_1 k})  \left( \Pr( \mathcal{A}_1 ) - H_2 e^{-\theta_2 k} \right) \\
& \geq \Pr( A < \beta k  )^{ck} - ( e^{-\theta_1 k} + H_2 e^{-\theta_2 k} - H_2 e^{- (\theta_1 + \theta_2)k}) \\
& \geq (1- \Pr(A \geq \beta k))^{ck} - H e^{-\theta k},
\intertext{where $\theta = \min(\theta_1, \theta_2)$ and $H>0$ such that $H < (1+H_2) $. Next, using $1-x \geq e^{-2 x}$ for small $x$, we have for all $k \geq k_0$}
\Pr( \underline{E}_1)  & \geq  e^{- 2 ck \Pr(A \geq \beta k)}  - H e^{-\theta k} \\
& \geq  1 - 2 ck \Pr(A \geq \beta k)  - H e^{-\theta k} \\
& \geq  e^{- 4 ck \Pr(A \geq \beta k)  - 2 H e^{-\theta k} }.
\end{align*}

Next, at time $\T_1 = T_1$, on event $\underline{E}_1$, the queue has at least $2k$ jobs, e.g. $Q_{\T_1} \geq 2k$, and no jobs have departed. Similarly as before, let $T_2 = \sum_{i=ck+1}^{3c k} A_i$ be the cumulative time that includes the next $2ck$ failures, and define $\mathcal{A}_2 \equiv \mathcal{A}_2(k) = \{ A_{ck +1} <2 \beta  k, A_{ck+ 2} < 2 \beta  k, \dots, A_{3ck} < 2 \beta  k \}$. The probability that no job departs in $(0, \T_2]$, where $\T_2 = T_1 + T_2$, is lower bounded by 
\begin{align}\label{eq:E1E2}
\nonumber \Pr(\text{no job departs in} (0,\T_2]) & \geq \Pr( Z_{T_1} \geq k, \A_1, Q_{\T_1} \geq 2k, Z_{(\T_1,\T_2]} \geq 2k, \A_2)\\
& \geq \Pr( Z_{T_1} \geq k, \A_1, Z_{T_2} \geq 2k, \A_2), 
\end{align}
since $\{ Q_{\T_1} \geq 2k \} \supseteq \{Z_{T_1} \geq k, \A_1 \}$ on the set $\{ Q_0 \geq k \}$. 

Now, if $E_2$ is the event that there is no departure in the next $2ck$ attempts and there are at least $2k$ arrivals in $(\T_1, \T_2]$, then $E_2 \supset \underline{E}_2 \eqdef \{ Z_{T_2} \geq 2 k, \mathcal{A}_2 \} $; note that $ \underline{E}_2$ is independent of $\underline{E}_1$. Via identical arguments as before, we obtain
\begin{align*}
\Pr( \underline{E}_2) & \geq \Pr(Z_{T_2} \geq 2k, T_2 \geq 4 k/ \lambda ,  \mathcal{A}_2) \\
& \geq \Pr(Z_{4k/ \lambda} \geq 2 k ) \left( \Pr( \mathcal{A}_2 ) - \Pr( T_2 < 4k/ \lambda) \right)\geq  e^{- 8 ck \Pr(A \geq 2 \beta k)  - 2 H e^{-2\theta  k} }.
\end{align*}
Therefore, at time $\T_2$, on event $\underline{E}_1 \cap \underline{E}_2$, there are at least $4k$ jobs.

In general, for any $n$, we can extend the reasoning from \eqref{eq:E1E2} to obtain
\begin{align*}
\Pr(\text{no job departs in} (0,\T_n]) & \geq \Pr( Z_{T_1} \geq k, \A_1, Z_{T_2} \geq 2k, \A_2, \dots, Z_{T_n} \geq 2^{n-1}k, \A_n)\\
& =  \Pr(\underline{E}_1   \cap \underline{E}_2 \cap \dots \cap \underline{E}_n),
\end{align*}
where $\underline{E}_n = \{ Z_{T_n} \geq 2^{n-1}k, \A_n \}$ and $T_n = \sum_{i= (2^{n-1}-1)ck+1}^{(2^{n}-1) c k} A_i$. Similarly, 
\begin{displaymath}
\Pr( \underline{E}_n)  \geq e^{- 2^{n+1} ck \Pr(A \geq 2^{n-1} \beta k)  - 2 H e^{-\theta 2^{n-1} k} }.
\end{displaymath}

Hence, we obtain
\begin{align*}
  \Pr(E_1  \cap E_2 \cap \dots \cap E_n) & \geq  \Pr( \underline{E}_1 \cap \underline{E}_2 \cap \dots \cap \underline{E}_n)  =  \Pr(\underline{E}_1) \Pr(\underline{E}_2) \cdots \Pr(\underline{E}_n),\\
   \intertext{since the events $\underline{E}_i$'s are independent. Thus,}
\Pr(E_1  \cap E_2 \cap \dots \cap E_n)   & \geq \prod_{i=1}^{n} e^{- 2^{i+1} ck \Pr(A \geq 2^{i-1} \beta k)  - 2 H e^{-2^{i-1} \theta k} }\\
& = e^{ -4 \sum_{i=0}^{n-1} 2^{i} c k \Pr(A \geq 2^{i} \beta k) - 2 H \sum_{i=0}^{n-1} e^{-2^i \theta k}}\\
& \geq e^{ -4 \sum_{i=0}^{\infty} 2^{i} c k \Pr(A \geq 2^{i} \beta k) - 2 H e^{-\theta k}\sum_{i=0}^{\infty} e^{-(2^i-1) \theta k}}. 
\intertext{Now, observe that $\sum_{i=0}^{\infty} e^{-(2^i-1) \theta k} < \infty$, and thus we can pick $H$ such that}
  \Pr(E_1  \cap E_2 \cap \dots \cap E_n)  & \geq  e^{ -4 \sum_{i=0}^{\infty} 2^{i} c k \Pr(A \geq 2^{i} \beta k) -  H e^{-\theta k}} .
\end{align*}

Furthermore, we observe that
\begin{align*}
  \sum_{i=0}^{\infty} 2^{i} c k \Pr(A \geq 2^{i} \beta k) &   \leq  \frac{c}{\beta}  \sum_{i=0}^{\infty}  \beta  k \int_{2^{i}}^{2^{i+1}} \Pr(A \geq x \beta k ) dx\\
    & \leq  \frac{c}{\beta} \beta  k   \int_{1}^{\infty} \Pr(A \geq x \beta k ) dx  =  \frac{c}{\beta}   \int_{\beta k }^{\infty} \Pr(A \geq y) dy  =  \frac{c}{\beta}   \E A \ind(A \geq \beta k ).
\end{align*}
and thus
\begin{align*}
  \Pr(E_1  \cap E_2 \cap \dots \cap E_n) &  \geq  e^{ -4 c \beta^{-1} \E A \ind(A \geq \beta k )  -  H e^{-\theta k}}  \geq 1 - H(\E A \ind(A \geq \beta k ) + e^{-\theta k})
  \end{align*}

Last, note that
\begin{align*}
  \Pr(\text{no job ever completes service}) \geq \Pr( \cap_{i=1}^{\infty} E_i) &= \lim_{n \rightarrow \infty}   \Pr(E_1  \cap E_2 \cap \dots \cap E_n) \\
  & \geq  1 - H(\E A \ind(A \geq \beta k ) + e^{-\theta k}), 
\end{align*}
where the first inequality follows by definition and the second equality from monotone convergence.

Last, for $k < k_0$, $\Pr(\text{no job ever completes service}| Q_0 \geq k) \geq 0 \geq 1 - H(\E A \ind(A \geq \beta k_0) + e^{-\theta k_0}) \geq 1 - H(\E A \ind(A \geq \beta k) + e^{-\theta k})$, by picking $H > 1/(\E A \ind(A \geq \beta k_0) + e^{-\theta k_0})$, and thus  \eqref{eq:prop1} holds trivially.  
\end{proof}

We proceed with our main theorem which shows that, after some finite time, no job will ever depart. Here, we can assume that the failure process is in stationarity, or more generally, that the first failure occurs at time $0 \leq A_0 < \infty$ a.s. and that the time between subsequent failures $i$ and $i+1$, is $A_i$, where $\{ A_i \}_{i \geq 1}$ are i.i.d., independent of $A_0$. If $A_0$ is the excess distribution of $A_1$, then the failure process is stationary.

\begin{theorem}\label{thm:0}
In the M/G/1 PS queue, if $\E A < \infty$ and $\Pr[B \geq \beta] = 1, \beta >0$, then
\begin{displaymath}
 \lim_{t \rightarrow \infty} \Pr (\text{no job ever completes service after time $t$}) =1. 
\end{displaymath}
\end{theorem}

\begin{proof}
For any $k \geq 1$, let $T_k$ be the first time that there are $k$ jobs in the queue and a failure occurs. $T_k$ is almost surely finite since it is upper bounded by the time $\bar{T}_k$ that there are at least $k$ arrivals in an open interval of size $\beta$ just before a failure. The probability of this event is $\Pr(Z_\beta \geq k) >0$.  

Let $\mathcal{B} \eqdef \{ B^{T_k}_1, \dots, B^{T_k}_{Q_{T_k}} \}$ denote the job sizes that are present in the queue at time $T_k$. From Proposition~\ref{prop:1}, we have 
\begin{align}\label{eq:thm1-a}
 \Pr(\text{no job leaves after $T_k$} | Q_{T_k}, \mathcal{B} )  & \geq 1 - H(\E A \ind(A \geq \beta k ) + e^{-\theta k}) \geq 1 - \epsilon,
 \end{align}
 for all $k \geq k_0$, since $\theta >0$ and $\E A \ind(A \geq \beta k ) \rightarrow 0$ as $k  \rightarrow \infty$. 

Now, for any fixed time $t$, we obtain
\begin{align*}
 \Pr(\text{no job leaves after time $t$}) & \geq \Pr(T_k \leq t, \text{no job leaves after $T_k$}) \\
 & = \E [  \Pr(T_k \leq t| Q_{T_k},  \mathcal{B} ) \Pr(\text{no job leaves after $T_k$} | Q_{T_k},  \mathcal{B} ) ]\\
 & \geq \Pr(T_k \leq t) (1 - \epsilon),
  \end{align*}
which follows from \eqref{eq:thm1-a}; the equality follows from the fact the event \{no job leaves after $T_k$\} is independent of the past, e.g. $T_k \leq t$, given $Q_{T_k},  \mathcal{B} $. Next, recall that $T_k$ is almost surely finite, i.e. $ \lim_{t \rightarrow \infty} \Pr(T_k \leq t) = 1 $, and thus taking the limit as $t \rightarrow \infty$ yields
\begin{align*}
 \lim_{t \rightarrow \infty}  \Pr(\text{no job leaves after time $t$}) & \geq 1 - \epsilon.
 \end{align*}
Last, letting $\epsilon \downarrow 0$ finishes the proof. 
\end{proof}

\begin{corollary}\label{cor:0}
Under the conditions in Theorem~\ref{thm:0}, we have as $t \uparrow \infty$,
\begin{displaymath}
Q_t  \uparrow  \infty \quad \text{a.s.} 
\end{displaymath}
\end{corollary}

\begin{proof}
Note that the number of arrivals $Z_t  \uparrow \infty$ as $t  \uparrow \infty$ a.s. Thus, without loss of generality, we can assume that $Z_t(\omega) \uparrow \infty$ as $t  \uparrow \infty$ for every $\omega$ (by excluding the set of zero probability). Then, for any $v >0$,
\begin{displaymath}
U_v \eqdef \{ \text{no job ever completes service after time $v$} \} \subset \{ Q_t  \uparrow  \infty \text{ as } t \uparrow \infty \}.
\end{displaymath}
Now, if $\omega \in U_v$, then for $t \geq v, Q_t(\omega) $ is non-decreasing. Furthermore, since there are no departures, the rate of increase of $Q_t $ is equal to the arrival rate, and thus $Q_t  \uparrow  \infty $. Hence, 
\begin{align*}
\Pr(Q_t  \uparrow  \infty  \text{ as } t \uparrow \infty) & \geq \Pr(\text{no job ever completes service after time $v$})  \\
\intertext{which, by Theorem~\ref{thm:0}, implies}
 \Pr(Q_t  \uparrow  \infty \text{ as } t \uparrow \infty) &= \lim_{v \rightarrow \infty} \Pr(\text{no job ever completes service after time $v$})  = 1.
 \end{align*} 
\end{proof}
\begin{remark}
Note that Theorem~\ref{thm:0} and Corollary~\ref{cor:0} are stronger than standard stability theorems, since they also imply that eventually no job ever leaves the system.
\end{remark}

Finally, we show instability, in general, without the condition $\Pr[B \geq \beta] = 1$. However, the conclusion is slightly weaker than in Theorem~\ref{thm:0}, and is the same as in Corollary~\ref{cor:0}. Basically, one cannot guarantee that no job ever completes service, since jobs can be arbitrarily small. 

\begin{theorem}\label{thm:0-b}
In the M/G/1 PS queue, if $\E A < \infty$ and $B >0$ a.s., we have as $t \uparrow \infty$,
\begin{displaymath}
Q_t  \uparrow  \infty \quad \text{a.s.} 
\end{displaymath}
\end{theorem}

\begin{remark}
Note that $B >0$ a.s. is just a non-triviality condition that excludes zero-sized jobs, i.e. non-existent ones.
\end{remark}

\begin{proof}
First, by assumption, we can pick $\beta>0$ such that $\Pr[B \geq \beta] >0$. Then, for any time $t$, let $ Q^\beta_t$ be the number of jobs whose size is at least $ \beta$, i.e. they satisfy $\Pr[B \geq \beta]=1$, and $q^\beta_t$ be the number of jobs that are smaller than $\beta$. Hence, 
\begin{align*}
Q_t =  Q^\beta_t + q^\beta_t \geq \underline Q^\beta_t,
\end{align*}
where $\underline Q^\beta_t$ is the queue in a system with the same arrival process where jobs of size $B \geq \beta$ are served in isolation. By Corollary~\ref{cor:0}, $\underline Q^\beta_t \uparrow  \infty$ a.s., and, therefore, we obtain $Q_t\uparrow  \infty$ a.s. 
\end{proof}

\subsubsection{Extension to DPS}
PS cannot capture the heterogeneity of users, which is associated with unequal sharing of resources. Hence, we discuss the \textit{Discriminatory} Processor Sharing (DPS) queue which is a multi-class generalization of the PS queue: all jobs are served simultaneously at rates that are determined by a set of weights $w_i, i = 1, \dots, K $. If there are $n_j$ jobs in class $j$, each class-$k$ job receives service at a rate $c_k = w_k/ \sum_{j=1}^K w_j n_j$. 

DPS has a broad range of applications. In computing, it is used to model Weighted-Round-Robin (WRR) scheduling. In communication networks, DPS is used for modeling heterogenous, e.g. with different round trip delays, TCP connections. Despite the fact that the PS queue is well understood, the analysis of DPS has proven to be very hard; yet, our previous result on PS is easily extended to DPS in the corollary below. 

\begin{corollary}
Under the conditions in Theorems~\ref{thm:0} and \ref{thm:0-b}, the discriminatory processor sharing (DPS) queue is also always unstable, with the same conclusion as in Theorems~\ref{thm:0} and \ref{thm:0-b}, respectively.
\end{corollary}

\begin{proof}
Without loss of generality, assume that the set of weights is ordered such that $w_1\leq w_2 \dots \leq w_K$. In the M/G/1 DPS queue, the service allocation at any given time $t$ for a single customer in class $k$ is given by 
\begin{align*}
c_k(t) = \frac{w_k}{\sum_{i=1}^K w_i n_i(t)} \leq \frac{ w_k}{w_1 \sum_{i=1}^K n_i(t)} \leq \frac{w_K  }{w_1 Q_t}. 
\end{align*}
Note that $c(t)  = w_K /(w_1 Q_t)    $ is the service rate in a PS queue with capacity $c = w_K /w_1 \geq  1$. Therefore, each class-$k$ job, $k= 1 \dots K$, in the DPS queue is served at a lower rate than the rate $ c$ of the PS queue. Hence, 
\begin{equation*}
Q^{DPS}_t \geq Q^{PS(c)}_t,
\end{equation*}
and since, under the conditions in Theorem~\ref{thm:0}, the PS queue is always unstable, it follows that the DPS queue is also unstable. 
\end{proof}

\subsection{Stability of First Come First Serve Queue}\label{ss:0-1}
In the FCFS discipline, each job is processed one at a time and therefore the expected service time for a single job is given in Definition~\ref{def:S} as 
\begin{align*}
 \E [S] &= \E \left[ \sum_{i=1}^{N-1} A_i + B\right].
 \intertext{Note that $N \eqdef \inf \{  n: A_n > B\} $ is a well defined stopping time for the process $(A, \{A_n\}_{n \geq 1})$, and thus the expected service time follows from Wald's identity as}
  \E [S] &= \E \left[ \sum_{i=1}^{N} A_i - A_N + B\right] \\
  &= \E [N] \E [A] - \E [A_N] + \E [B].  
 \end{align*}
Now, assuming that the availability periods $A$ are exponentially distributed with rate $\mu$ (Poisson failures), the expected service time is given by 
\begin{align}\label{eq:ES}
\nonumber \E [S] &  = \E [N] \E [A] - (\E[ A] + \E[B] ) + \E[B]  \\
& = (\E [N] - 1) \E [A],
\end{align}
since $\E[A_N] =\E \left[ \E[A| A>B]  \right]= \E[A + B] = \E [A] + \E [B]$, due to the memoryless property of the exponential distribution. 

The necessary and sufficient condition for the stability of the M/G/1 FCFS queue is 
\begin{align*}
\lambda \E [S] < 1. 
\end{align*}
Now, let the jobs be fixed and all equal to some positive constant $\beta >0$. Since $A$ is exponentially distributed with rate $\mu$, then 
\begin{align} \label{eq:EN}
\nonumber \Pr[N > n] &= \Pr(A \leq \beta )^n = G(\beta )^n,
\intertext{and thus, the expected number of restarts is}
 \E [N] = \sum_{n=0}^\infty \Pr[N>n] &= \sum_{n=0}^\infty G(\beta)^n = \bar G(\beta)^{-1} = e^{\mu  \beta }.
 \end{align}
 
Furthermore, for fixed jobs $B = \beta$, we can compute explicitly $\E [S]$ without the exponential assumption on $A$. To this end, note that
\begin{align*}
\E [S] &  = \E \left[ \sum_{i=1}^{N-1} A_i + \beta  \right] =  \E \left[ \sum_{n=2}^{\infty} \mathbf{1}_{\{ N = n \}}  \sum_{i=1}^{n-1} A_i + \beta  \right] \\
& =   \E \left[ \sum_{n=2}^{\infty} \mathbf{1}_{\{ A_1 < \beta, A_2 < \beta, \dots, A_{n-1}< \beta, A_n \geq \beta \}}  \sum_{i=1}^{n-1} A_i   \right] + \beta  \\
& = \sum_{n=2}^{\infty} \E  \left(  \sum_{i=1}^{n-1} A_i \mathbf{1}_{\{ A_1 < \beta, A_2 < \beta, \dots, A_{n-1}< \beta, A_n \geq \beta \}}  \right) + \beta,  \\
\intertext{and since $(A,\{A_i\}_{i \geq 1}) $ are i.i.d, we obtain}
\E [S]  & = \sum_{n=2}^{\infty} (n-1)  \E \left[ A \mathbf{1}_{\{ A < \beta\}} \right]  \Pr( A < \beta)^{n-2} \Pr(A \geq \beta )  + \beta  \\
& = \sum_{n=2}^{\infty} (n-1)  \E \left[ A \mathbf{1}_{\{ A < \beta\}} \right]  \Pr( N= n-1)  + \beta,
\end{align*} 
where we recall that $\Pr(N=n) = \Pr(A<\beta)^{n-1} \Pr(A \geq \beta) $, by definition, and thus,
\begin{align}\label{eq:ES2}
\nonumber \E [S]   & =  \E \left[ A \mathbf{1}_{\{ A < \beta\}} \right]   \sum_{n=1}^{\infty} n \Pr( N= n)  + \beta \\
 & =  \E \left[ A \mathbf{1}_{\{ A < \beta\}} \right]  \E [N] + \beta.
\end{align} 
 
Hence, for exponential $A$, the preceding expression (or \eqref{eq:ES}) yields
\begin{align*}
\E [S] &=  (e^{\mu \beta} - 1) \mu^{-1},
\end{align*}
and the stability region reduces to 
\begin{equation}\label{eq:stab}
\lambda \E[S] = \lambda \mu^{-1} (e^{\mu \beta} - 1) < 1.
\end{equation} 

Note that $\E [S] \geq \mu^{-1} \mu \beta = \beta $, where $\lambda \beta <1$ gives the stability region of the ordinary M/G/1 queue without failures. We observe that as the jobs grow in size, the stability region shrinks. In other words, the larger the $\beta$, the slower the arrival rate the queue can accommodate. For FCFS scheduling, if it is not possible to adjust the arrival rate, we could potentially decrease the job sizes, e.g. apply fragmentation techniques, in order to maintain a large stability region without generating too much overhead, resulting from dividing a single job into many smaller ones.  

\section{GI/G/1 PS queue with restarts}\label{s:1}
In the previous section, we show that PS is unstable assuming Poisson arrivals. Here, we show that this result can be further generalized to more general arrival distributions, e.g. renewal process. However, to avoid technical complications we assume that the failure process is Poisson, i.e. the availability periods $A_i$ are exponential. To this end, we use $M_{(t_0, t_1]}$ to denote the number of Poisson failures in $(t_0,t_1]$ and write $M_t$ for intervals of the form $(0,t]$. Let $( \tau,  \{ \tau_n \}_{n \geq 1} )$ be an i.i.d. sequence, where $\tau_n$ represent the interarrival times of the renewal process. 

The main purpose of this section is to show that there is nothing special about the Poisson arrival assumption that leads to instability. Instead, the instability results from the interplay between sharing and retransmission/restart mechanisms. First, we prove the following proposition using similar arguments as in Proposition~\ref{prop:1}.

\begin{proposition}\label{prop:2}
Assume that at time $t=0$, a new job arrives and there are $Q_0 \geq k$ jobs in the GI/G/1 PS queue. If $\E A < \infty$, $\E \tau^{1+\delta} < \infty, 0 < \delta <1$ and $\Pr [B \geq \beta] = 1, \beta >0$, then for all $k \geq 1$
\begin{equation} \label{eq:1a}
\Pr[\text{no job ever completes service}] \geq 1 - O(\E A \ind (A \geq \beta k ) +  k^{- \delta}).
\end{equation}
\end{proposition}

\begin{proof}
Let $T_1 = \sum_{i=1}^{k} \tau_i$ be the cumulative time that includes the first $k$ arrivals and $M_{T_1}$ be the number of failures in $(0,T_1)$. Now, define the event $\mathcal{A}_1 \equiv \mathcal{A}_1(k) \eqdef \{ A_1 < \beta  k, A_2 <  \beta  k, \dots, A_{M_{T_1}} < \beta  k \}$. On this event, no job can leave the system since $Q_0 \geq k$ and all of them are at least of size $\beta$. Thus, if they were served in isolation, they could not have completed service in the first $M_{T_1}$ attempts.

Now, let $E_1$ denote the event that there is no departure in the first $M_{T_1}$ attempts and there are at most $ck$ failures in $(0, T_1]$. Formally, 
\begin{equation*}
E_1 \supset \underline{E}_1 \eqdef \{ M_{T_1} \leq ck, \mathcal{A}_1 \},
\end{equation*}
on the set $\{ Q_0 \geq k\}$. Now, observe that
\begin{align*}
 \Pr( \underline{E}_1) & \geq \Pr(M_{T_1} \leq c k,   A_1 < \beta  k, A_2 <  \beta  k, \dots, A_{M_{T_1}} < \beta  k ) \\
  & \geq \Pr(M_{T_1} \leq c k,   A_1 < \beta  k, A_2 <  \beta  k, \dots, A_{ck} < \beta  k ) \\
& \geq \Pr( A_1 < \beta k )^{ck} - \Pr( M_{T_1} > c k ).
\end{align*}

Next, note that
\begin{align*}
 \Pr(M_{T_1} > c k  ) &  = \Pr \left(M_{T_1} > c k  , T_1 \leq   \frac{3 k  \E \tau}{2} \right) + \Pr \left(M_{T_1} > c k, T_1 > \frac{ 3 k \E \tau}{ 2} \right) \\
 & \leq \Pr \left( M_{ \frac{ 3 k \E \tau }{2}  } > c k  \right)  + \Pr \left(T_1 >  \frac{ 3 k \E \tau }{ 2} \right),
 \intertext{where the first term is negligible for $c > 2 \lambda \E \tau $ since the expected number of failures is $3 k \lambda \E \tau /2  $. Now, observe that}
  \Pr(T_1 >  \frac{3 k \E \tau}{2}) & =  \Pr \left( \sum_{i=1}^{k} \tau_i > \frac{ 3k \E \tau}{2} \right) = \Pr \left(\sum_{i=1}^{k} (\tau_i - \E \tau) > \frac{3k \E \tau}{2} - k \E \tau \right).
 \intertext{Now, let $X_i  \eqdef \tau_i - \E \tau $, which are bounded from above since $ \E \tau  < \infty$, from our main assumption. Therefore, by choosing $h = 2^{-\delta} (\E \tau )^{1+\delta} $ and $y = \E \tau/ 4$ in Lemma~1 of \cite{JT2007}, we obtain}
  \Pr \left(\sum_{i=1}^{k} X_i  >  k \E \tau  /2  \right) &  \leq k \Pr( X_1 >  k \E \tau/ 4)  + \frac{ h k} {2^{-\delta} (k \E \tau)^{1+\delta} } \\ 
  & \leq k \Pr( \tau_1 >  k \E \tau /4 + \E \tau)  + \frac{1} { k^{\delta} }\\
  & \leq k \frac{\E \tau^{1+\delta}}{(k \E \tau/4 +  \E \tau)^{1+\delta}} + k^{-\delta}   \leq 2  k^{-\delta}.
\end{align*}

Therefore, 
 \begin{align*}
\Pr( \underline{E}_1)  &  \geq    (1- \Pr(A \geq \beta k))^{ck} -  2  k^{-\delta},
\intertext{where using $1-x \geq e^{-2 x}$ for small $x$, we have for all $k \geq k_0$}
\Pr( \underline{E}_1)   \geq  e^{- 2 ck \Pr(A \geq \beta k)}  -  2  k^{-\delta} & \geq  1 - 2 ck \Pr(A \geq \beta k)  -  2  k^{-\delta} \\
& \geq  e^{- 4 ck \Pr(A \geq \beta k)  -  4  k^{-\delta}}.
\end{align*}

Next, at time $\T_1 = T_1$, on event $\underline{E}_1$, the queue has at least $2k$ jobs, e.g. $Q_{\T_1} \geq 2k$, and no jobs have departed. Similarly as before, let $T_2 = \sum_{i=k}^{3 k} \tau_i$ be the cumulative time that includes the next $2k$ arrivals, and define $\mathcal{A}_2 \equiv \mathcal{A}_2(k) = \{ A_{M_{T_1} +1} <2 \beta  k, A_{ck+ 2} < 2 \beta  k, \dots, A_{M_{T_1 + T_2}} < 2 \beta  k \}$. The probability that no job departs in $(0, \T_2]$, where $\T_2 = T_1 + T_2$, is lower bounded by 
\begin{align}\label{eq:E1E2-b}
\nonumber \Pr(\text{no job departs in} (0,\T_2]) & \geq \Pr( M_{T_1} \leq ck, \A_1, Q_{\T_1} \geq 2k, M_{(\T_1,\T_2]} \leq 2ck, \A_2)\\
& \geq \Pr( M_{T_1} \leq ck, \A_1, M_{(\T_1,\T_2]} \leq 2ck, \A_2), 
\end{align}
since $\{ Q_{\T_1} \geq 2k \} \supseteq \{M_{T_1} \leq ck, \A_1 \}$ on the set $\{ Q_0 \geq k \}$. 

Now, if $E_2$ is the event that there is no departure in the next $M_{T_2}$ attempts and there are at most $2ck$ failures in $(\T_1, \T_2]$, then $E_2 \supset \underline{E}_2 \eqdef \{ M_{T_2} \leq 2 c k, \mathcal{A}_2 \} $; note that $ \underline{E}_2$ is independent of $\underline{E}_1$ due to Poisson memoryless property. Via identical arguments as before, we obtain
\begin{align*}
\Pr( \underline{E}_2) & \geq \Pr(M_{T_2} \leq 2ck, A_{ck+1} < \beta k, \dots, A_{3ck}< \beta k) \\
& \geq  e^{- 8 ck \Pr(A \geq 2 \beta k)  - 4 (2k)^{-\delta} }.
\end{align*}
Therefore, at time $\T_2$, on event $\underline{E}_1 \cap \underline{E}_2$, there are at least $4k$ jobs.

In general, for any $n$, we can extend the reasoning from \eqref{eq:E1E2-b} to obtain
\begin{align*}
\Pr(\text{no job departs in} (0,\T_n]) & \geq \Pr( M_{T_1} \leq ck, \A_1, M_{T_2} \leq 2ck, \A_2, \dots, M_{T_n} \leq 2^{n-1}k, \A_n)\\
& =  \Pr(\underline{E}_1   \cap \underline{E}_2 \cap \dots \cap \underline{E}_n),
\end{align*}
where $\underline{E}_n = \{ M_{T_n} \leq 2^{n-1}c k, \A_n \}$ and $T_n = \sum_{i= (2^{n-1}-1)k+1}^{(2^{n}-1) k} \tau_i$. Similarly, 
\begin{displaymath}
\Pr( \underline{E}_n)  \geq e^{- 2^{n+1} ck \Pr(A \geq 2^{n-1} \beta k)  - 4( 2^{n-1} k)^{-\delta} }.
\end{displaymath}

Hence, we obtain
\begin{align*}
  \Pr(E_1  \cap E_2 \cap \dots \cap E_n) & \geq \prod_{i=1}^{n} e^{- 2^{i+1} ck \Pr(A \geq 2^{i-1} \beta k)  - 4( 2^{i-1} k)^{-\delta} }\\
& = e^{ -4 \sum_{i=0}^{n-1} 2^{i} c k \Pr(A \geq 2^{i} \beta k) - 4 k^{-\delta}  \sum_{i=0}^{n-1} ( 2^{i})^{-\delta}}\\
& \geq e^{ -4 \sum_{i=0}^{\infty} 2^{i} c k \Pr(A \geq 2^{i} \beta k) -4 k^{-\delta}  \sum_{i=0}^{\infty} 2^{-\delta i} }. 
\intertext{Now, observe that $\sum_{i=0}^{\infty} 2^{-\delta i}< \infty$, and thus we can pick $H >0 $ such that}
  \Pr(E_1  \cap E_2 \cap \dots \cap E_n)  & \geq  e^{ -4 \sum_{i=0}^{\infty} 2^{i} c k \Pr(A \geq 2^{i} \beta k) -  H k^{-\delta}} .
\end{align*}

The remainder of the proof follows identical arguments as Proposition~\ref{prop:1}. Thus, 
\begin{align*}
  \Pr(\text{no job ever completes service})   \geq  1 - H(\E A \ind(A \geq \beta k ) +  k^{-\delta}).
\end{align*} 
\end{proof}

\begin{theorem}\label{thm:0b}
In the GI/G/1 PS queue, if $\E A < \infty, \E \tau^{1+\delta}, 0 <\delta <1$ and $\Pr[B \geq \beta] = 1, \beta >0$, then
\begin{displaymath}
 \lim_{t \rightarrow \infty} \Pr (\text{no job ever completes service after time $t$}) =1. 
\end{displaymath}
\end{theorem}

\begin{proof}
Similarly as in the proof of Theorem~\ref{thm:0}, we observe the system at time $T_k$ when there are $k$ jobs in the queue and a failure occurs. Since the arrivals are non Poisson, we need additional reasoning to ensure that $ T_k < \infty$ a.s. In this regard, let $t_k$ be a large enough time interval such that with positive probability it includes at least $k$ arrivals. Next, we split this interval into smaller ones of size $\beta$ and require that there is a failure in each of those intervals. Since the failures are Poisson, this event has a positive, albeit extremely small, probability. Hence, the first time $T_k$ that the queue has at least $k$ jobs and a failure occurs is a.s. finite. 

Now, the remainder of the proof follows the same arguments as in Theorem~\ref{thm:0} of Section~\ref{s:0}. We omit the details. 
\end{proof}
Similarly as in Theorem~\ref{thm:0-b} of Section~\ref{s:0}, we drop the condition $\Pr[B \geq \beta] = 1$ and prove general instability. 

\begin{theorem}\label{thm:1b}
In the GI/G/1 PS queue, if $\E A < \infty, \E \tau^{1+\delta}, 0 <\delta <1$ and $B >0$ a.s., we have as $t \uparrow \infty$,
\begin{displaymath}
Q_t  \uparrow  \infty \quad \text{a.s.} 
\end{displaymath}
\end{theorem}

The proof is similar to the proof of Theorem~\ref{thm:0-b} and thus is omitted. Furthermore, the equivalent results could be stated for the DPS scheduler as well. Last, the preceding findings could be further extended to both non Poisson arrivals and non Poisson failures. However, the proofs would be much more involved and complicated; here, we avoid such technicalities. 


\section{Transient Behavior - Scheduling a Finite Number of Jobs}\label{s:3}
In the previous sections, we focus on the steady state behavior of the M/G/1 queue with restarts and prove that PS is always unstable for failure distributions with finite first moment. We also show instability for the GI/G/1 queue, assuming Poisson failures. In this section, in order to gain further insight into this system, we study its transient behavior. In this regard, we consider a queue with a finite number of jobs and no future arrivals and compute the total time until all jobs are completed. In Subsections~\ref{s:3-1} and \ref{s:3-2}, we analyze the system performance when the jobs are served one at a time and when Processor Sharing (PS) is used, respectively. More precisely, for a finite number of jobs with sizes $B_i, 1 \leq i \leq m$, and assuming no future arrivals, we study the completion time $\Theta_m$, until all $m$ jobs complete their service.

Note that in the case of traditional work conserving scheduling systems the completion time does not depend on the scheduling discipline and is always simply equal to $ \sum_{i=1}^m B_i$. However, in channels with failures there can be a stark difference in the total completion time depending on the scheduling policy. This difference can be so large that in some systems the expected completion time can be infinite while in others finite, or even having many high moments. 

Overall, we discover that, with respect to the distribution of the total completion time $\Theta_m$, serving one job at a time exhibits uniformly better performance than PS; see Theorems~\ref{thm:2} and \ref{thm:3}. Furthermore, when the hazard functions of the job and failure distributions are proportional, i.e. $\log \bar F(x) \sim \alpha \log \bar G (x)$, we show that PS performs distinctly worse for the light-tailed job/failure distributions as opposed to the heavy-tailed ones, see parts (i) and (ii) of Theorem~\ref{thm:3}.

Before presenting our main results, we state the following theorem on the logarithmic asymptotics of the time $\bar S =\sum_{i=1}^N A_i = S + (A_N - B)$, where $S$ is from Definition~\ref{def:S}. Note that $\bar S$ includes the remaining time $(A_N - B)$ until the next channel availability period, thus representing a natural upper bound for $S$. In the following, let $\vee \equiv \max$.

\begin{theorem}\label{thm:1}
If $\log \bar F(x) \sim \alpha \log \bar G (x)$ for all $x \geq 0$ and $\alpha >1$, and $\E [B^{\alpha + \delta}] < \infty, \E [A^{1\vee \alpha}] < \infty$ for some $\delta >0$, then
\begin{equation}\label{eq:1}
\lim_{t \rightarrow \infty}\frac{\log \Pr[ \bar S> t ] }{\log t} = - \alpha  \qquad \text{as } t \rightarrow \infty.
\end{equation} 
\end{theorem}

\begin{proof}
By Theorem~6 in \cite{JT2007}, when specialized to the conditions of this theorem, we obtain that $ \log \Pr[S>t] \rightarrow - \alpha \log t$ as $t \rightarrow \infty$. This immediately yields the lower bound for $\bar S =   S + ( A_N - B) \geq S$. For the upper bound, $\bar S =  S + ( A_N - B) $ and the union bound result in 
\begin{equation*}
\Pr [ \bar S > 2 x] \leq \Pr [S > x]  + \Pr [ A_N - B > x] .
\end{equation*}
 Hence, in view of Theorem~6 in \cite{JT2007}, we only need to bound $\Pr[A_N - B > x]$. To this end, observe that
\begin{align*}
\Pr[A_N - B > x] &= \Pr[A_N > B + x ] = \sum_{i = 1 }^{\infty} \Pr[A_i > B + x, N = i] \\
& = \sum_{i = 1 }^{\infty} \Pr[A_i > B + x, A_1 < B, \dots, A_{i-1} < B] \\
& = \sum_{i = 1 }^{\infty} \E  \left[ \Pr \left( A_i > B + x| B \right) \Pr \left( A_1 < B|B \right)^{i-1} \right] \\
& = \E   \left[ \frac {\bar G ( B + x)  } {\bar G (B) }  \right] \leq \bar G (x) \E [N],
\end{align*}
since $\E[N] = \E (1/\bar G (B) )$. Now, the condition $\alpha > 1$ guarantees that $\E[N] < \infty$ whereas $\E [A^{\alpha }]$ implies that $\bar G(x)  = O (1/x^{\alpha  })$. Thus, \eqref{eq:1} is satisfied. 
\end{proof}


\subsection{Serving One Job at a Time}\label{s:3-1}
In this subsection, we consider the failure-prone system that was introduced in Section~\ref{s:intro}, with unit capacity. The jobs are served one at a time, e.g. First Come First Serve (FCFS). Herein, we analyze the performance of this system assuming that, initially, there are $m$ jobs in the queue and there are no future arrivals. Specifically, we study the total completion time, which is defined below.

\begin{definition}\label{def:Thetam}
The total completion time is defined as the total time until all the jobs are successfully completed and is denoted as
 \begin{displaymath}
\Theta_m \eqdef \sum_{i=1}^{m} S_i, 
\end{displaymath}
where $m$ is the total number of jobs in the system and $S_i$ is the service requirement for each job. 
\end{definition}
Next, we define the forward recurrence time, i.e. the elapsed time between some fixed $t_0$ until the time that the next failure occurs after $t_0$. 

\begin{definition} \label{def:tau}
Let $L_t$ be the number of failures in the interval $(0,t)$, i.e. the number of regenerative points of the renewal process $(A, \{A_n\}_{n\geq1})$. 
The forward recurrence time, which corresponds to the elapsed time until the next failure after time $t$, is defined as
\begin{equation}
\tau_t : = \sum_{n=1}^{L_t + 1} A_n - t .
\end{equation}
\end{definition}

In the following theorem, we prove that the tail asymptotics of the total completion time, from Definition~\ref{def:Thetam}, under this policy is a power law of the same index as the service time of a single job.

\begin{theorem}\label{thm:2}
If $\log \bar F(x) \sim \alpha \log \bar G (x)$ for all $x \geq 0$ and $\alpha >1$, and $\E [B^{\alpha + \delta}] < \infty, \E [A^{1\vee \alpha}] < \infty$ for some $\delta >0$, then
\begin{equation*}
\lim_{t \rightarrow \infty}\frac{\log \Pr[\Theta_m > t ] }{\log t} = -\alpha.
\end{equation*}
\end{theorem}

\begin{proof}
Recall that the service requirement for a job $B_i$ was previously defined as $S_i =\sum_{j=1}^{N_i-1} A_j + B_i$.

For the \emph{lower} bound, we observe that
\begin{align}\label{eq:1A}
\nonumber \Pr [\Theta_m > t ] & \geq \Pr[ S_1 > t], \\
 \intertext{since the total completion time is at least equal to the service time of a single job. By taking the logarithm and using Theorem~6 in \cite{JT2007}, we have }
 \frac{\log \Pr[\Theta_m  > t ]}{\log t} & \geq - (1+\epsilon)  \alpha.
\end{align}
 
For the \emph{upper} bound, we compare $\Theta_m$ with the completion time in a system where the server is kept idle between the completion time of the previous job and the next failure. Clearly, 
\begin{align}\label{eq:FCFS-ub}
\Theta_m \leq \bar \Theta_m \eqdef \sum_{i=1}^m \bar S_i,
\end{align}
where $  \bar S_i \eqdef \sum_{j=1}^{N_i} A_j $ are the service times that include the remaining availability period $A_{N_i}$.
We prove this intuitive claim more formally by induction in the appendix.

Then, we argue that
\begin{align}\label{eq:1B}
\nonumber \Pr[\Theta_m  > t ] & \leq \Pr \left[ \sum_{i=1}^m \bar S_i > t \right]  \leq m \Pr \left[  \bar S_1 > \frac{t}{m} \right],\\
\intertext{which follows from the union bound. By taking the logarithm and using Theorem~\ref{thm:1}, we have }
 \frac{\log \Pr[\Theta_m  > t ]}{\log t} & \leq - \alpha(1-\epsilon) + \frac{\log m}{\log t} \leq  -(1-2\epsilon) \alpha,
\end{align}
where we pick $t$ large enough such that $ \log t \geq  \log m /(\alpha \epsilon) $.

Letting $\epsilon \rightarrow 0$ in both \eqref{eq:1A} and \eqref{eq:1B} finishes the proof. 
\end{proof}

\subsection{Processor Sharing Discipline}\label{s:3-2}
In this subsection, we analyze the Processor Sharing discipline where $m$ jobs share the (unit) capacity of a single server. We present our main theorem on the logarithmic scale, which shows that the tail asymptotics of the total completion time is determined by the shortest job in the queue. In particular, under our main assumptions, this time is a power law, but it exhibits a different exponent depending on the job size distribution. 
\begin{itemize}
\item If the jobs are subexponential (heavy-tailed) or exponential, the total delay is simply determined by the time required for any single job to complete its service, as if it was the only one present in the queue.  
\item If the jobs are superexponential (light-tailed), the total delay is determined by the service time of the \emph{shortest} job. This job generates the heaviest asymptotics among all the rest.
\end{itemize}
Our main result, stated in Theorem~\ref{thm:3} below, shows that on the logarithmic scale the distribution of the total completion time $\Theta_m^{PS}$ is heavier by a factor $m^{\gamma-1}$ for superexponential jobs relative to the subexponential or exponential case. Therefore, in systems with failures and restarts, sharing the capacity among light-tailed jobs induces long delays, whereas, for heavy-tailed ones, PS appears to perform as good as serving the jobs one at a time. Interestingly enough, this deterioration in performance is determined by the time it takes to serve the shortest job in the system. 

Note that the in a PS queue with no future arrivals, the shortest job will depart first. Immediately after this, the server will continue serving the remaining $m-1$ jobs, and, similarly, the shortest job, i.e. the second shortest among the original $m$ jobs, will depart before all the others. This pattern will continue until the departure of the largest job, which is served alone.

\begin{theorem} \label{thm:3}
Assume that the hazard function $-\log \bar F(x)$ is regularly varying with index $\gamma \geq 0$. If $\log \bar F(x) \sim \alpha \log \bar G (x)$ for all $x \geq 0$ and $\alpha >1$, and $\E [B^{\alpha + \delta}] < \infty, \E [A^{1\vee \alpha}] < \infty$ for some $\delta >0$, then 
\begin{enumerate}[label=(\roman{*})]
\item if $\gamma \leq 1$, i.e. $B$ is subexponential or exponential, then 
\begin{equation*}
  \lim_{t \rightarrow \infty}\frac{-\log \Pr[\Theta_m^{PS} > t ] }{\log t} = \alpha,
\end{equation*}
\item if $\gamma > 1$, i.e. $B$ is superexponential, then
\begin{equation*}
\lim_{t \rightarrow \infty}\frac{-\log \Pr[\Theta_m^{PS} > t ] }{\log t} = \frac{\alpha}{ m^{\gamma-1}}< \alpha.
\end{equation*}
\end{enumerate}
\end{theorem}

\begin{remark} 
When $\alpha>1$, we easily verify that $\E[\Theta_m^{PS}] < \infty$ in case $(i)$, and the system is stable; if the jobs are superexponential, e.g. case $(ii)$, then $\E[\Theta_m^{PS}] = \infty$ if $\alpha < m^{\gamma-1}$. 
\end{remark}

\begin{proof}
Let $B^{(1)} \leq B^{(2)} \leq \dots \leq B^{(m)}$ be the order statistics of the jobs $B_1, B_2, \dots, B_m$. 

The assumption that $- \log \bar F(x)$ is regularly varying with index $\gamma$ implies that 
\begin{align}\label{eq:ma1}
\log \bar F(\lambda x)  \sim  \lambda^\gamma \log \bar F(  x) ,
\end{align}
for any $\lambda >0$. 

We begin with the \emph{lower} bound. 

\emph{(i) Subexponential or exponential jobs ($\gamma \leq 1$)}. \\
The total completion time is lower bounded by the time required for a single job to depart when it is exclusively served, e.g. if the total capacity of the system is used. Hence, it follows that
\begin{align} \label{eq:lb}
\Pr[ \Theta^{PS}_m > t] & \geq \Pr[S_1 > t ],
\end{align}
where $S_1$ is the service time of a single job of random size $B_1$, when there are no other jobs in the system. Now, recalling Theorem~6 in \cite{JT2007}, it holds that
\begin{align*}
\lim_{t \rightarrow \infty} \frac{\log \Pr[S_1 > t]}{ \log t} = -  \alpha.
\end{align*}
By taking the logarithm in \eqref{eq:lb}, the lower bound follows immediately.

\emph{(ii) Superexponential jobs ($ \gamma > 1)$}. \\
The total completion time is lower bounded by the delay experienced by the shortest job, and hence,
\begin{align} \label{eq:S1-}
\Pr[ \Theta^{PS}_m > t] & \geq \Pr[ S_1^{PS} > t ], 
\end{align}
where $ S_1^{PS}$ is the service time of job $B^{(1)}$. First, note that the distribution of $B^{(1)}$ is given by 
\begin{align}\label{eq:F1}
\nonumber \Pr(B^{(1)} > x) &= \Pr(B_1 >x, B_2 >x, \dots, B_m >x) \\
\nonumber & = \Pr(B_1 >x) \Pr( B_2 >x)  \cdots \Pr(B_m >x)\\
& = \Pr(B_1 > x )^m = \bar F(x)^m,
\end{align}
since $B_i, i = 1, \dots, m$, are independent and identically distributed. 
Now, the service time $S_1^{PS}$ is determined by the number of failures this job has experienced, i.e. 
\begin{align*}
\Pr[ N_1 > n] & = \E  \left[ \Pr \left( B^{(1)} > \frac{A}{m}  \right) \right]^n = \E \left (1-\bar G  (m B^{(1)}) \right)^n  , 
\end{align*}
and, using \eqref{eq:F1} and \eqref{eq:ma1}, together with our main assumption, we observe that
\begin{align*}
 \log  \Pr(m B^{(1)} > x ) & = m \log  \bar F \left( \frac{x}{m}  \right)  \\
 &  \sim m^{1-\gamma} \log  \bar F(x) \sim \alpha m^{1-\gamma} \log  \bar G(x).
\end{align*} 
Then, Theorem~6 in \cite{JT2007} applies with $ \alpha/m^{\gamma-1} \leq \alpha $, i.e.
\begin{align*}
\lim_{t \rightarrow \infty} \frac{\log \Pr[S_1^{PS} > t]}{ \log t} = - \frac{\alpha}{m^{\gamma-1}}.
\end{align*}

Next, we derive the \emph{upper} bound. To this end, we consider a system where the server is kept idle after the completion of each job until the next failure occurs. At this time, all the remaining jobs are served under PS until the next shortest one departs. If there are more than one jobs of the same size, only one of these departs. Under this policy, it clearly holds that
\begin{align*}
\Theta_m^{PS} \leq  \sum_{i=1}^m \bar S_i^{PS},
\end{align*}
where $\bar S_i^{PS}$ corresponds to the service time of the $i^{th}$ smallest job and includes the time until the next failure.

Using the union bound, we obtain
\begin{align}\label{eq:UB1}
\Pr[ \Theta^{PS}_m  > t] &\leq \Pr \left[\sum_{i=1}^m \bar S_i^{PS} > t \right] \leq (1+\epsilon) \sum_{i=1}^m \Pr \left(\bar S_i^{PS} > \frac{t}{m} \right).
\end{align}
It is easy to see that the service time of the $i^{th}$ smallest job $B^{(i)}$ depends on the number of jobs that share the server, i.e. $m-i +1$, since $m-i$ jobs have remained in the queue. 
Now, the distribution of the $i^{th}$ shortest job is derived as
\begin{align}\label{eq:Fi}
\nonumber \Pr(B^{(i)} > x) &= \sum_{k=0}^{i-1}  {m \choose k} \Pr(B_1\leq x)^k \Pr(B_1> x)^{m-k} \\
& \sim {m \choose i-1}  \Pr(B_1 > x)^{m-i+1} \sim \bar F(x)^{m-i+1}.
\end{align}

The number of restarts for the $i^{th}$ smallest job, $N_i$, is computed as 
\begin{align*}
\Pr[N_i > n ] &=  \E \left[ \Pr  \left( B^{(i)}> \frac{A}{m-i} \right) \right]^n \\
& =  \E \left(1-\bar G ((m-i+1) B^{(i)} ) \right )^n.  
\end{align*}
Next, starting from \eqref{eq:Fi}, it easily follows that
\begin{align*}
 \log  \Pr \left( (m-i+1) B^{(i)} > x \right) & \sim  \log    \bar F \left( \frac{x}{m-i+1}  \right)^{m-i +1}  \\
 & \sim  (m-i+1)^{1-\gamma}  \log \bar F(x) \\
 &  \sim \alpha  (m-i+1)^{1-\gamma}   \log \bar G(x) , 
\end{align*}
where we use \eqref{eq:ma1} and our main assumption and define $\alpha_i \eqdef \alpha/(m-i+1)^{\gamma-1} $.

Now, recalling Theorem~\ref{thm:1}, we have 
\begin{align*}
\frac{\log \Pr[\bar S_i^{PS} > t]}{\log t} \rightarrow \alpha_i \text{ as $ t\rightarrow \infty$,}
\end{align*}
and thus \eqref{eq:UB1} yields
\begin{equation*}
\frac{\log \Pr[\Theta_m^{PS} > t]}{  \log t} \leq  -(1 - \epsilon) \min_{i = 1 \dots m} \alpha_i,
\end{equation*}
for all  $t\geq t_0$.

\emph{(i) Subexponential or exponential jobs ($\gamma \leq 1$)}. \\
Observe that $\underset {i = 1\dots m } \min \alpha_i = \alpha $, and thus 
\begin{align}\label{eq:UBi}
\frac{\log \Pr[\Theta_m^{PS} > t]}{\log t} & \leq  - (1 - \epsilon)   \alpha.
\end{align}

\emph{(ii) Superexponential jobs ($\gamma >1$)}. \\
In this case, $\underset {i = 1\dots m } \min  \alpha_i = \alpha / m^{\gamma -1}$, and thus 
\begin{align}\label{eq:UBii}
\frac{\log \Pr[\Theta_m^{PS} > t]}{\log t} \leq    -(1-\epsilon) \frac{ \alpha}{ m^{\gamma -1 }}.
\end{align}
Letting $\epsilon \rightarrow 0$ in \eqref{eq:UBi} and \eqref{eq:UBii}, we obtain the upper bound. 
\end{proof}

 \begin{figure}[h]
  \centering
\includegraphics[width = 0.6 \textwidth]{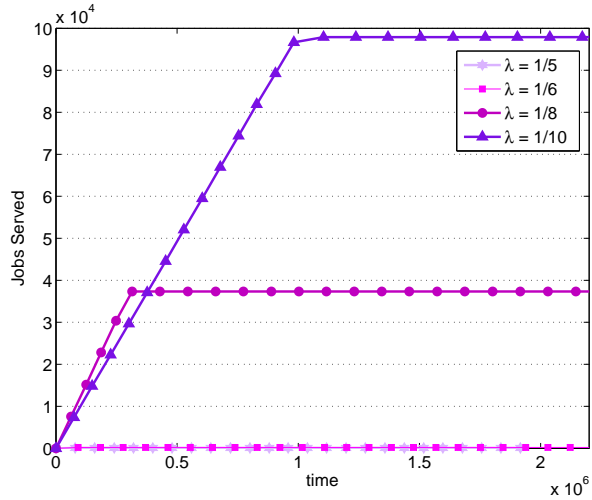}
  \caption{Example 1. Jobs served over time.}
\label{fig:ex1}
\end{figure}

\section{Simulation}\label{s:5}
In this section, we present our simulation experiments in order to demonstrate our theoretical findings. All the experiments result from $N=10^8$ (or more) samples of each simulated scenario; this guarantees the existence of at least 100 occurrences in the lightest end of the tail that is presented in the figures. 
First, we illustrate the instability results from Sections~\ref{s:0} and \ref{s:1}.

\textbf{Example 1.} \textit{M/G/1 PS is unstable.} In this example, we show that the PS queue becomes unstable by simulating the M/G/1 PS queue for different arrival rates $\lambda>0$, which all satisfy the stability condition for the M/G/1 FCFS queue, when jobs are served one at a time. In this regard, we assume constant job sizes $\beta = 1$ and Poisson failures of rate $\mu = 1/20$. Therefore, by evaluating \eqref{eq:stab}, we obtain
\begin{align*}
\lambda \E[S] = \lambda \mu^{-1} (e^{\mu} - 1) =  20(e^{0.05}-1) \lambda = 1.025 \lambda < 1,
\end{align*}
or equivalently the stability region for the FCFS queue is given by $\Lambda = \{ \lambda \leq 0 : \lambda < 0.9752 \}$. Hence, in this example, we use $\lambda$ from the FCFS stability region, $\lambda \in \Lambda$.

In Fig.~\ref{fig:ex1}, we plot the number of jobs that have received service up to time $t$. We observe that the cumulative number of served jobs always converges to a fixed number and does not increase any further. This happens after some critical time when the queue starts to grow continuously until it becomes unable to drain. For larger values of $\lambda$, the system saturates faster meaning that the cumulative throughput at the saturated state is lower.

\begin{figure}[h]
\begin{minipage}[t]{0.5\textwidth}
\centering
  \includegraphics[width =  \textwidth]{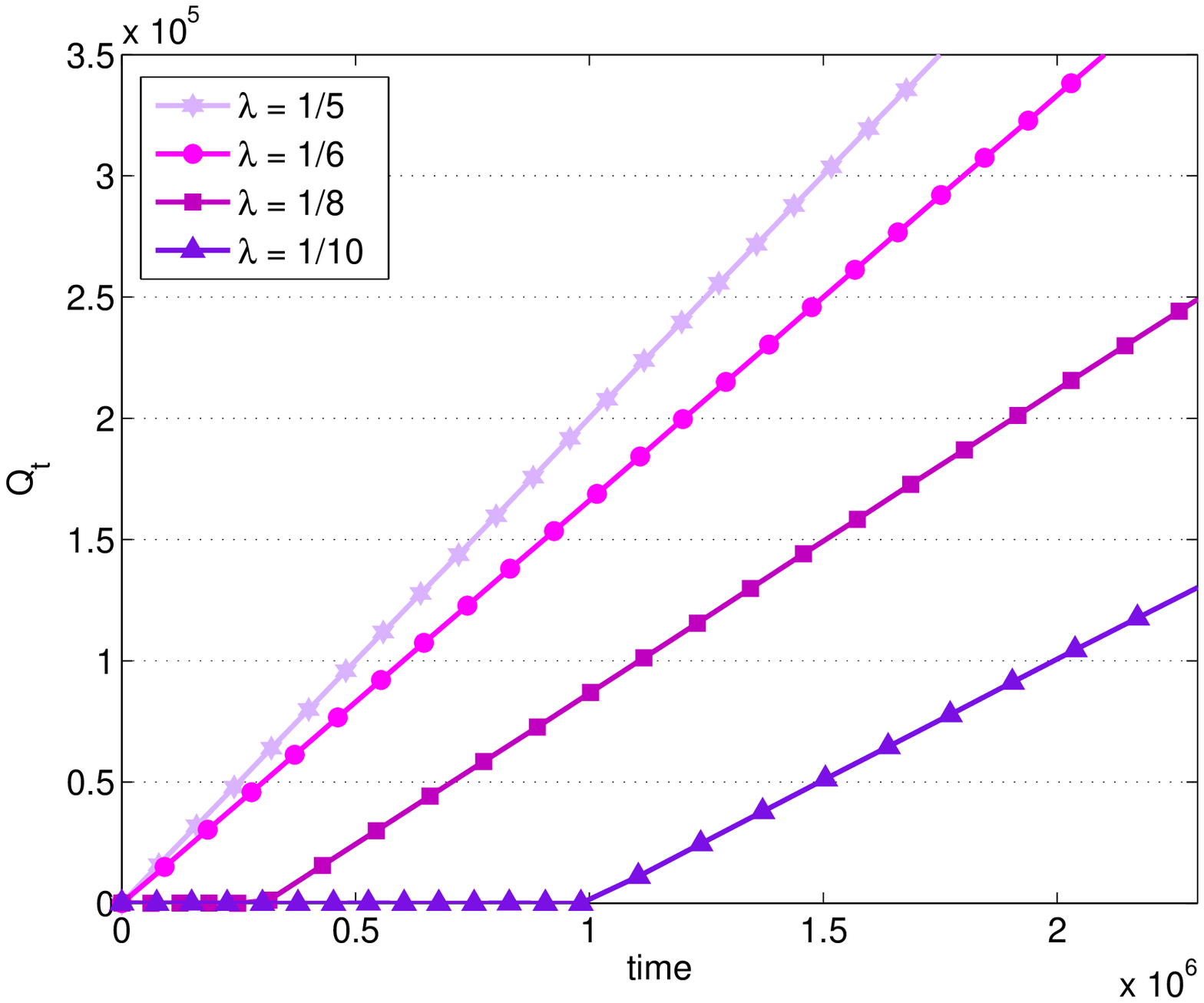}
  (a) Queue size over time
 \end{minipage} 
  \begin{minipage}[t]{0.5\textwidth}
    \centering
\includegraphics[width =  \textwidth]{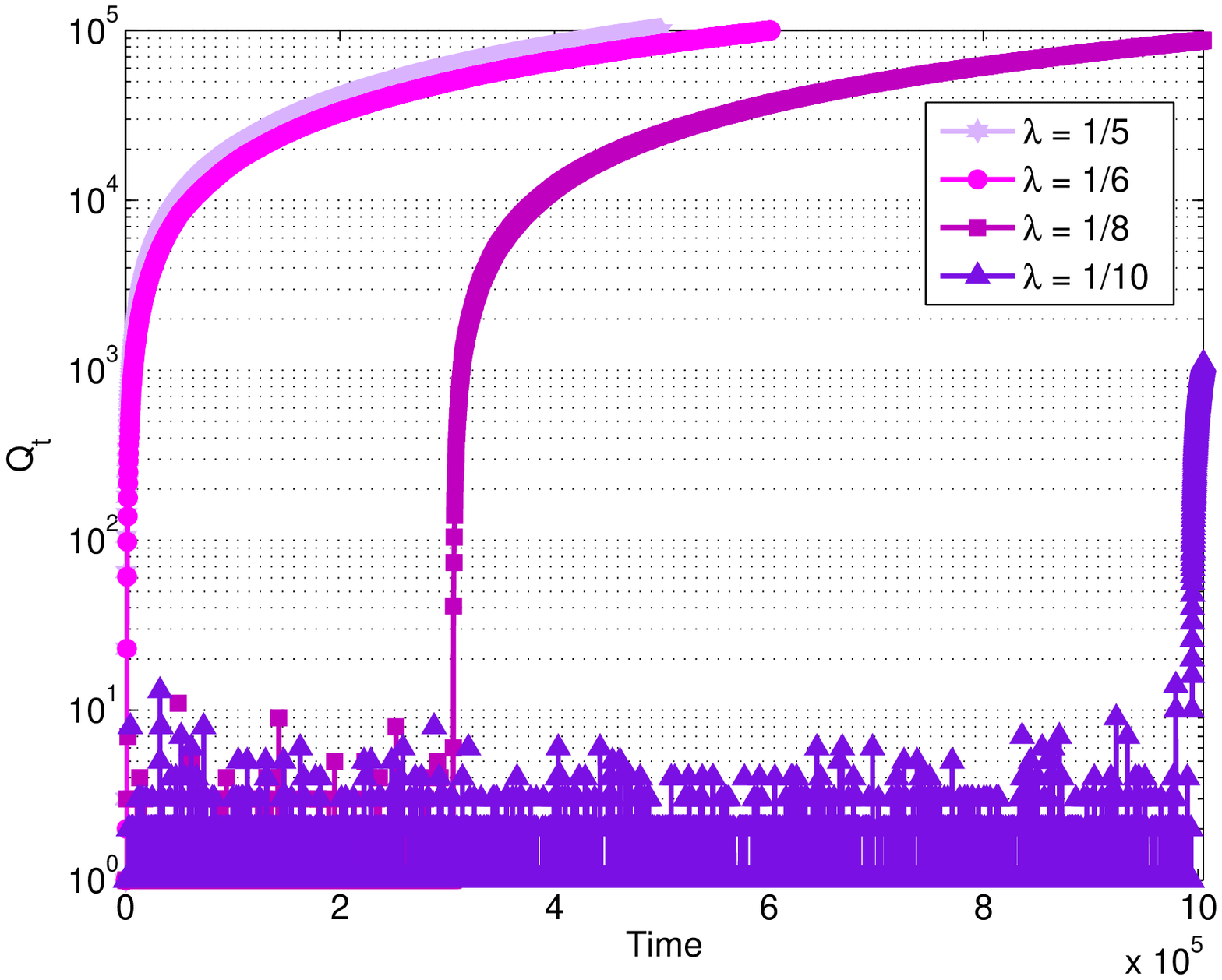}
  \label{fig:ex1c}
  (b) Queue size for small $t$
    \end{minipage}
      \caption{Example 1. Queue size evolution. \small Subfigure (b) zooms in the time range $[0, 10^6]$ of Fig.~\ref{fig:ex1b}; $Q_t$ ($y$-axis) is shown on the logarithmic scale.}
       \label{fig:ex1b}
    \end{figure}

Furthermore, we observe from the simulation that the system behaves as if it were stable until some critical time or queue size after which it is unable to drain. From Fig.~\ref{fig:ex1}, we can see that the case $\lambda = 10^{-1}$ saturates at time $t = 10^6$ and the total number of served jobs reaches $10^5$. Hence, the departure rate until saturation time is $10^5/10^6= 10^{-1}$, which is exactly equal to the arrival rate $\lambda = 10^{-1}$, corresponding to the departure rate of a stable queue. This further emphasizes the importance of studying the stability of these systems since, at first glance, they may appear stable.

 \begin{figure}[h]
  \centering
\includegraphics[width = 0.6 \textwidth]{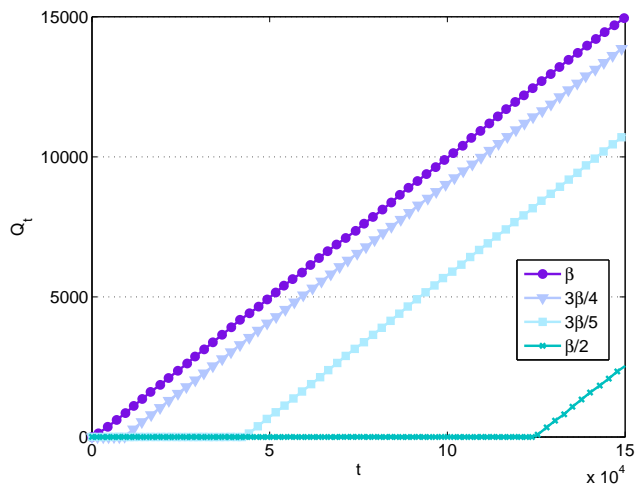}
  \caption{Example 1. Queue size over time parameterized by fragment length; $\beta = 2, \lambda = 0.1$.}
\label{fig:ex1d}
\end{figure}

Fig.~\ref{fig:ex1b} demonstrates the queue size evolution over time. Similarly as in Fig.~\ref{fig:ex1}, we observe that for any arrival rate $\lambda$, there is a critical time after which the queue continues to grow and never empties. This time varies depending on the simulation experiment; yet, on average, we observe that the queue remains stable for longer time when $\lambda$ is smaller. Now, we zoom in on the queue evolution on the logarithmic scale in Fig.~\ref{fig:ex1b}(b). Again, we observe that the queue looks stable until some critical time/queue size.

Last, in Fig.~\ref{fig:ex1d}, we plot the queue evolution for different job sizes, namely $\beta = 1, 1.2, 1.5$ and 2. We observe that larger job fragments cause instability much faster than the smaller units. For example, $\beta = 2$ leads to instability almost immediately, while $ \beta = 1.5$ renders the queue unstable after $10^4$ time units. Similarly, reducing the fragment size by $60$\% delays the process by an additional $ 3\times 10^4$ units. Last, cutting the jobs in half causes instability after approximately $13 \times 10^4$ time units. This implies that one should apply fragmentation with caution in order to select the appropriate fragment size that will maintain good system performance for the longest time.

\begin{figure}[h]
  \centering
\includegraphics[width = 0.6 \textwidth]{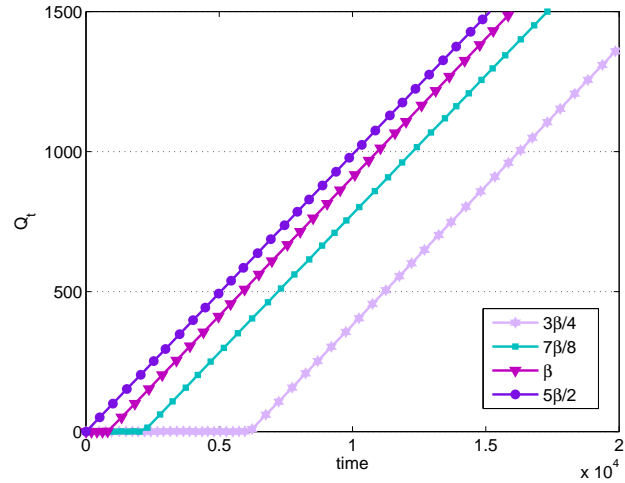}
  \caption{Example 2. Queue size over time parameterized by job size; $\beta =4$.}
\label{fig:ex2}
\end{figure}

\textbf{Example 2.} \textit{General arrivals.} In this example, we consider non Poisson arrivals. We assume that the failure distribution is exponential with mean $\E A = 10$ and that jobs interarrival times follow the Pareto distribution with $\alpha = 2$ and mean $\E \tau = 10.1$. Similarly as in the previous example, Fig.~\ref{fig:ex2} shows the queue evolution with time for different job sizes $\beta$.  

Next, we validate the results on the transient analysis from Section~\ref{s:3}.

 \textbf{Example 3.} \textit{Serving one job at a time/FCFS: Always the same index $\alpha$}. In this example, we consider a queue of $m=10$ jobs, which are served First Come First Serve (FCFS), i.e. one at a time. The logarithmic asymptotics from Theorem~\ref{thm:2} implies that the tail is always a power law of index $\alpha =2$.

 \begin{figure}[h]
\centering
\includegraphics[width = 0.6 \textwidth]{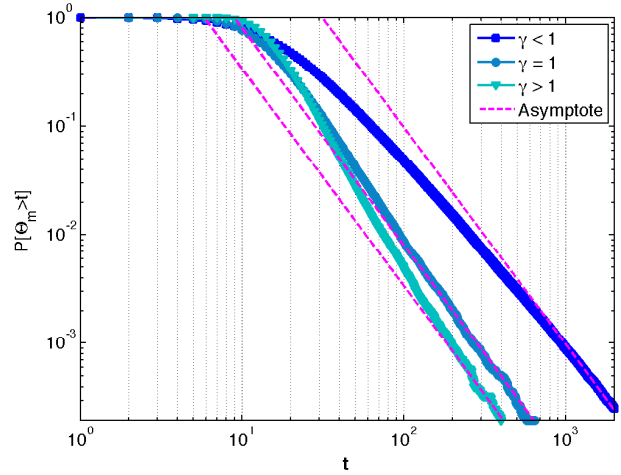}
  \caption{Example 3. FCFS: Logarithmic asymptotics when $\alpha = 2$ for exponential, superexponential ($\gamma > 1$) and subexponential ($\gamma < 1$) distributions.}
\label{fig:ex3}
\end{figure}

In Fig.~\ref{fig:ex3}, we plot the distribution of the total completion time in a queue with 10 jobs that are processed one at a time. On the same graph, we plot the logarithmic asymptotics (dotted lines) that correspond to a power law of index $\alpha =2$. We consider the following three scenarios:
\begin{enumerate} 
\item Weibull distributions with $\gamma = 2$. The failures $A$ are distributed according to $\bar G(x) = e^{-(x/\mu)^2}$ with mean $\E [A] = \mu \Gamma(1.5) =  1.5$, and jobs $B$ also follow Weibull distributions with $\bar F(x) = e^{-(x/\lambda)^2}, \lambda = \mu / \sqrt 2  $. In this case, it is easy to check that the main assumption of Theorem~\ref{thm:2} is satisfied, i.e. $$\log \bar F(x) = - (x/\lambda)^2 = \alpha \log \bar G(x), \alpha  = (\mu/ \lambda)^2. $$ 
\item Exponential distributions. $A$'s are exponential with $\E [A] = 2$, $\bar G(x) = e^{-x/2}$, and the jobs $B$ are also exponential of unit mean, i.e. $\bar F(x) = e^{-x}$. Then, trivially, $$\log \bar F(x) =  2 \log \bar G(x).$$ 
\item Weibull distributions with $\gamma =  0.5$. $A$'s are Weibull with $\bar G(x) = e^{-\sqrt{x}/2}$, i.e $\E [A] = 8$. Also, we assume Weibull jobs $B$ with $\bar F(x) = e^{-\sqrt{x}}$. Thus,
 $$\log \bar F(x) = - \sqrt{x} = 2  \log \bar G(x). $$
\end{enumerate}
In all three cases, we obtain $\alpha  = 2$. Yet, we observe that the tail asymptotics is the same regardless of the distribution of the job sizes. For the subexponential jobs (case 3: Weibull with $\gamma < 1$), the power law tail appears later compared to the case of superexponential jobs. This is because the constant factor of the exact asymptotics is different for each case, and it depends on the mean size of $A$, $\E [A]$.  

\begin{figure}[h]
 \centering
\includegraphics[width =  0.6 \textwidth]{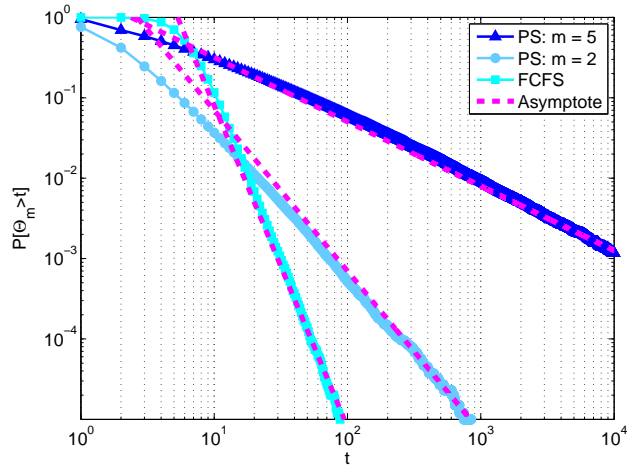}
 \caption{Example 4. Logarithmic asymptotics for different number of superexponential jobs when $\alpha = 4$ under PS and FCFS discipline.}
\label{fig:ex4a}
\end{figure}

\textbf{Example 4.} \textit{PS: The effect of the number of jobs}. In this example, we consider a PS queue with $m=5$ and $m=2$ superexponential jobs, and compare it against a FCFS queue with $m=5$ jobs. We assume superexponential job sizes $B$'s and $A$'s, namely Weibull with $\gamma = 2$; see case 1 of Example~3. Here $\alpha$ is taken equal to 4. The logarithmic asymptotics is given in Theorems~\ref{thm:2} and \ref{thm:3}. 

In Fig.~\ref{fig:ex4a}, we demonstrate the total completion time $\Theta_m^{PS}$, for different number of jobs, when $\gamma = 2$. Theorem~\ref{thm:3}(ii) states that $\alpha(m) = \alpha/m^{\gamma-1} $ and, thus, for $\gamma =2$ we have $\alpha(m) = \alpha / m$, e.g. we expect power law asymptotes with index $ \alpha/ m$ for the different values of $m$. On the same figure, we also plot the FCFS completion time $\Theta_m$, which is always a power law of index $\alpha = 4$, as we previously observed in Example~3. It can be seen that PS generates heavier power laws, for superexponential jobs. In particular, PS with $m=2$ results in power law asymptotics with $\alpha(2)=2$, while PS with $m=5$ jobs leads to system instability since $\alpha(5) = 4/5< 1$. 

 \begin{figure}[h]
  \centering
\includegraphics[width = 0.6 \textwidth]{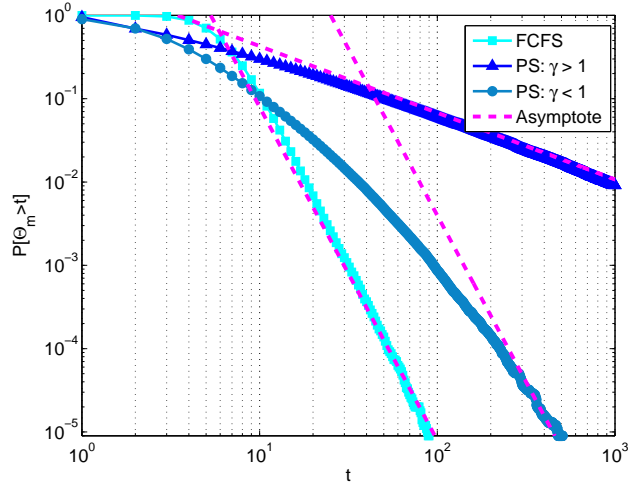}
  \caption{Example 5. Logarithmic asymptotics under FCFS, PS with subexponential and superexponential jobs.}
\label{fig:ex5}
\end{figure}

\textbf{Example 5.} \textit{PS: The effect of the distribution type}. In this example, for completeness, we evaluate the impact of the job distribution on the total completion time under both heavy and light-tailed job sizes. To this end, we consider the PS queue from Example~4, with $m=5$ jobs, and compare it against FCFS. In Fig.~\ref{fig:ex5}, we re-plot the logarithmic asymptotics of the total completion time $\Pr(\Theta_m^{PS}>t)$, for different distribution types of the failures/jobs and index $\alpha = 4$, as before. In particular, we consider Weibull distributions as in Example~3 with $\gamma = 1/2 < 1$ and $\gamma = 2>1$ for the subexponential and superexponential cases, respectively.

On the same graph, we plot the distribution of the completion time $\Theta_m$ in FCFS, which is always a power law of the same index, as illustrated in Example~3. By fixing the number of jobs to be $m=5$, Fig.~\ref{fig:ex5} shows that when the jobs are superexponential, PS yields the heaviest asymptotics among all three scenarios; for subexponential jobs, PS generates asymptotics with the same power law index as in FCFS, albeit with a different constant factor.

 \begin{figure}[h]
  \centering
\includegraphics[width = 0.6 \textwidth]{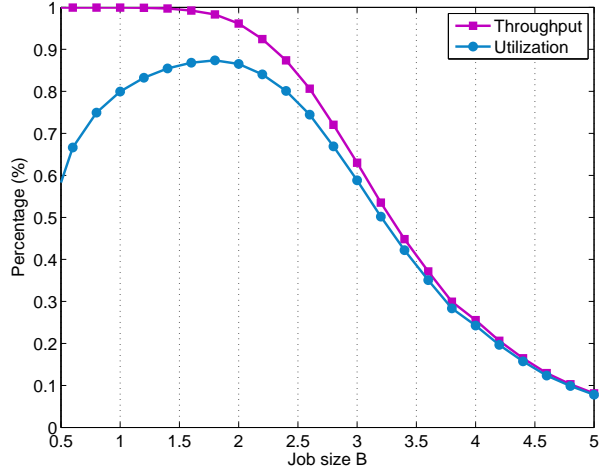}
  \caption{Example 6. Throughput vs. utilization tradeoff.}
\label{fig:ex6}
\end{figure}

\textbf{Example 6.} \textit{Limited queue: Throughput vs. overhead tradeoff}. In practice, job and buffer sizes are bounded and therefore the queue may never become unstable. However, our results indicate that the queue may lock itself in a `nearly unstable' state, where it is at its maximum size and the throughput is very low. Here, we would like to emphasize that, unlike in the case of unlimited queue size, job fragmentation can be useful in increasing the throughput and the efficiency of the system. In this case, one has to be careful about the overhead cost of fragmentation. Basically, each fragment requires additional information, called the `header' in the context of communications, which contains details on how it fits into the bigger job, e.g. destination/routing information in communication networks. Hence, if the fragments are too small, there will be a lot of overhead and waste of resources. In view of this fact, one would like to optimize the fragment sizes by striking a balance between throughput and utilization.

In this example, we demonstrate the tradeoff between throughput and generated overhead, assuming limited queue size $\tilde Q$. If the newly arriving job does not fit in the queue, i.e. the number of jobs currently in the queue is equal to $\tilde Q$, it is discarded. We define throughput as the percentage of the jobs that complete service among all jobs that arrive at the M/G/1 PS queue. It basically corresponds to the total work that is carried out in the system. On the other hand, we define utilization as the useful work that is served over the aggregated load in the system. Specifically, we consider jobs that require a minimum size $\beta$, where $\beta$ represents the overhead, e.g. the packet header, thread id, etc. The remaining job size, $B-\beta$, represents the useful information. 

We consider different job sizes $B$ from 0.4 up to 5 bytes, with overhead $\beta = 0.2$. We simulate the M/G/1 PS queue with maximum queue size $\tilde Q= 10$ jobs for a fixed time $T=10^{8}$ time units. The arrivals are Poisson with rate $1/10$ and the failures are exponential of the same rate. Clearly, in the case of fixed job sizes $B$, throughput $\gamma$ is lower bounded by the throughput of the system when it performs at the limit, i.e. when the queue is full. This state corresponds to the worst overall performance and can be easily computed. On average, for a fixed period of time $T$, $\tilde Q$ jobs will complete service every $\E[S_{\tilde Q}]$ time units, while the total jobs that arrive in the system is $\lambda T$. In this case, the lower bound for the throughput is given by $\min \{ 1, \gamma \}$, where
\begin{equation*}
\underline{\gamma} =  \tilde Q \frac {T }{\E[S_{\tilde Q}]} \frac{1} {\lambda T} = \frac{\tilde Q}{\lambda \E[S_{\tilde Q}]},
\end{equation*} 
and in the particular case of exponential failures, using \eqref{eq:stab} we derive 
\begin{equation}
\underline{\gamma} =  \frac{\tilde Q}{\lambda \mu^{-1}(e^{\mu \tilde Q B} -1 ) }.
\end{equation} 

Using this observation, throughput will be suboptimal when $\gamma < 1$. Thus, for job sizes larger than $B_* =\log ( \mu \tilde Q  \lambda^{-1}+ 1)/ (\mu \tilde Q)$, the throughput starts decreasing.

In Fig.~\ref{fig:ex6}, we observe that for small job sizes, the throughput is 100\% and it deteriorates as the job size $B$ increases. In particular, when the job size exceeds $1.5$, the throughput drops exponentially. Utilization exhibits a different behavior; it is low when the job size is small, i.e. the useful job size is comparable to the overhead $\beta$, and reaches its peak at $B \approx 1.7$. After this, it starts decreasing following similar trend as the throughput. In this case, $B-\beta \approx 1.5$ appears to be the optimal size for the job fragments. This phenomenon of combining limited queue size with job fragmentation may require further investigation.

\section{Concluding Remarks}\label{s:6}
Retransmissions/restarts represent a primary failure recovery mechanism in large-scale engineering systems, as it was argued in the introduction. In communication networks, retransmissions lie at the core of the network architecture, as they appear in all layers of the protocol stack. Similarly, PS/DPS based scheduling mechanisms, due to their inherent fairness, are commonly used in computing and communication systems. Such mechanisms allow for efficient and fair resource allocation, and thus they are preferred in engineering system design.

However, our results show that, under mild conditions, PS/DPS scheduling in systems with retransmissions is always unstable. Furthermore, this instability cannot be resolved by job fragmentation techniques or checkpointing. On the contrary, serving one job at a time, e.g. FCFS, can be stable and its performance can be further enhanced with fragmentation. Interestingly, systems where jobs are served one at a time can highly benefit from fragmentation and, in fact, their performance can approach closely the corresponding system without failures.

Overall, using PS in combination with retransmissions in the presence of failures deteriorates the system performance and induces instability. In addition, our findings suggest that further examination of existing techniques is necessary in the failure-prone environment with retransmission/restart failure recovery and sharing, e.g. see Example~6.

\section*{Appendix}
\begin{proof}[of \eqref{eq:FCFS-ub} in Theorem~\ref{thm:2}]
We formally prove that $\Theta_m \leq  \bar{\Theta}_m = \sum_{i = 1}^m \bar S_i  $. 
The proof follows by induction. 

$\mathit{n = 1}.$ Let $\tau_{\Theta_1}$ denote the time between $\Theta_1$ until the first failure occurs after $\Theta_1$, i.e. the time after the departure of the first job (see also Definition~\ref{def:tau}). Then the total service time for jobs $B_1$ and $B_2$ is
\begin{align*}
\Theta_2 = \left \{ \begin{array}{l l}
\Theta_1 + \tau_{\Theta_1}+ S_2, &  \text{if $B_2>  \tau_{\Theta_1}$} \\
\Theta_1 + B_2 ,&  \text{otherwise.}
\end{array} \right.
\end{align*}
If we idle the server after the successful completion of the first job until the next failure, the total completion time will be equal to $\bar{\Theta}_2 = \Theta_1 +  \tau_{\Theta_1} + S_2$, since we discard the remaining interval $\tau_{\Theta_1}$ and start service at $\Theta_1 + \tau_{\Theta_1}$. Therefore,  $\Theta_2 \leq  \bar{\Theta} _2$ (see Figure~\ref{fig:appendix} for an illustration).

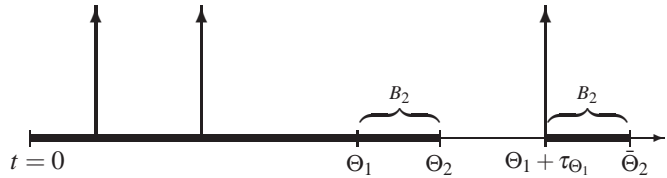
\begin{figure}[h] 
\centering
\begin{picture}(300,85)
\put(0,25){\vector(1,0){240}}
\put(0,22){\line(0,1){6}}
\put(124,22){\line(0,1){6}}
\put(195,22){\line(0,1){6}}
\put(227,22){\line(0,1){6}}
\put(155,22){\line(0,1){6}}
\thicklines
 \put(0,25){\linethickness{0.85mm}\line(1,0){125}}
\put(3,13){\makebox(0,0)[b]{$t = 0$}}
\put(65,25){\vector(0,1){50}}
\put(25,25){\vector(0,1){50}}
\put(140,30){\makebox(0,0)[b]{$ \overbrace{\qquad \quad}^{B_2}$}}
\put(125,11){\makebox(0,0)[b]{\small $ \Theta_1$}}
\put(195,25){\vector(0,1){50}}
 \put(125,25){\linethickness{0.85mm}\line(1,0){30}}
\put(155,11){\makebox(0,0)[b]{\small $\Theta_2$}}
\put(196,11){\makebox(0,0)[b]{\small $\Theta_1 + \tau_{\Theta_1}$}}
 \put(195,25){\linethickness{0.85mm}\line(1,0){32}}
 \put(211,30){\makebox(0,0)[b]{$ \overbrace{\qquad \quad}^{B_2}$}}
\put(229,11){\makebox(0,0)[b]{\small $ \bar \Theta_2$}}
\end{picture}
\caption{Completion time in a failure-prone system: \small{Assume that there are two jobs $B_1$ and $B_2$ and the first succeeds at time $\Theta_1$. In the original system, job $B_2$ starts service immediately and completes at time $\Theta_2$, before the next failure occurs after $\tau_{\Theta_1}$ time units. In the alternate system, $B_2$ will only start its service at time $\Theta_1 + \tau_{\Theta_1}$. If $B_2 < \tau_{\Theta_2}$, then $\bar \Theta_2=\Theta_1 + \tau_{\Theta_1}+ B_2$.}}
\label{fig:appendix}
\end{figure}

\textit{Induction step.} Assume $\Theta_n \leq \bar \Theta_n$ for $n < m$. If $\Theta_n $ is the time when the $n^{th}$ job is completed, then $\tau_{\Theta_n}$ is the time until the next failure after $\Theta_n$. Now, for the following job $B_{n+1}$, we have  
\begin{align*}
\Theta_{n+1} = \left \{ \begin{array}{l l}
\Theta_n +  \tau_{\Theta_n} + S_{n+1}, &  \text{if $B_{n+1} >  \tau_{\Theta_n}$} \\
\Theta_n + B_{n+1}, &  \text{otherwise.} \\
\end{array} \right.
\end{align*}
If $\Theta_n   = \bar \Theta_n $, then clearly $\bar \Theta_{n+1}  = \Theta_n + \tau_{\Theta_n} +  S_{n+1} \geq \Theta_{n+1}  $. 

If $\bar \Theta_n > \Theta_n   $ then, by construction, job $B_{n+1}$ can start its service after time $ \Theta_n +  \tau_{\Theta_n} $, i.e the time that the first failure occurs after $\Theta_n$. This implies that $\bar \Theta_n \geq \Theta_n +  \tau_{\Theta_n}  $. Now, if $B_{n+1}$ finishes before the failure occurs, then clearly $\bar \Theta_{n+1} \geq  \Theta_{n+1} $. If not, it will either succeed during the period $(\Theta_n +  \tau_{\Theta_n}, \bar \Theta_{n})$ implying that $  \Theta_{n+1} \leq  \bar \Theta_{n+1} $, or it will synchronize with the other system and $\Theta_{n+1}  =  \bar \Theta_n +  \tau_{\bar \Theta_n} + S_{n+1} \leq \bar \Theta_{n+1}  $. 
\end{proof}


\end{document}